\documentclass{article}

\PassOptionsToPackage{numbers, compress}{natbib}

\usepackage[preprint]{neurips_2025}

\usepackage{comment}
\usepackage{amsmath}
\usepackage{thm-restate}
\usepackage{amssymb, amsthm}
\usepackage{mathtools}
\usepackage{thmtools}
\usepackage{bm}
\usepackage{tikz}
\usepackage{dsfont}
\usepackage{wrapfig}
\usetikzlibrary{shapes.geometric, positioning}

\usepackage[utf8]{inputenc} 
\usepackage[T1]{fontenc}    
\usepackage{hyperref}  
\usepackage{cleveref}
\usepackage{url}            
\usepackage{booktabs}       
\usepackage{amsfonts}       
\usepackage{nicefrac}       
\usepackage{microtype}      
\usepackage{xcolor}  
\usepackage[linesnumbered,ruled,vlined]{algorithm2e}

\newtheorem{definition}{Definition}
\newtheorem*{proposition*}{Proposition}
\newtheorem{lemma}{Lemma}

\DeclareMathOperator*{\argmax}{arg\,max}
\usepackage{aliases}
\usepackage{natbib}
\newcommand{\mar}[1]{\textcolor{blue}{Maria : #1}}
\newcommand{\VP}[1]{\textcolor{magenta}{VP : #1}}

\title{Online matching on stochastic block model}

\author{
  Maria Cherifa \\
  Université de Paris, MAP5, Paris, France\\
  CREST, ENSAE, Palaiseau, France\\
  Faiplay joint team\\
  \texttt{maria.cherifa@ensae.fr} \\
   \And
   Clément Calauzènes \\
   Criteo AI Lab, Paris, France  \\
    Fairplay joint team  \\
   \texttt{c.calauzenes@criteo.com} \\
   \AND
   Vianney Perchet \\
  CREST, ENSAE, Palaiseau, France \\
  Criteo AI Lab, Paris, France\\
  Fairplay joint team\\
\texttt{vianney.perchet@normalesup.org} \\
}

\begin{document}

\maketitle

\begin{abstract}
While online bipartite matching has gained significant attention in recent years, existing analyses in stochastic settings fail to capture the performance of algorithms on heterogeneous graphs, such as those incorporating inter-group affinities or other social network structures. In this work, we address this gap by studying online bipartite matching within the stochastic block model (SBM). A fixed set of offline nodes is matched to a stream of online arrivals, with connections governed probabilistically by latent class memberships. We analyze two natural algorithms: a $\myopic$ policy that greedily matches each arrival to the most compatible class, and the $\balance$ algorithm, which accounts for both compatibility and remaining capacity. For the $\myopic$ algorithm, we prove that the size of the matching converges, with high probability, to the solution of an ordinary differential equation (ODE), for which we provide a tractable approximation along with explicit error bounds. For the $\balance$ algorithm, we demonstrate convergence of the matching size to a differential inclusion and derive an explicit limiting solution. Lastly, we explore the impact of estimating the connection probabilities between classes online, which introduces an exploration-exploitation trade-off. 
\end{abstract}

\section*{Introduction}
Finding matchings in bipartite graphs is a fundamental problem at the intersection of graph theory~\cite{Godsil1981MatchingsAW,Zdeborova_2006}, network science, and combinatorial optimization~\cite{lovasz2009matching,schrijver2003combinatorial}. In this context, a bipartite graph $ \mathcal{G} = (\mathcal{N}, \mathcal{T}, \mathcal{E}) $ is defined by two disjoint node sets, $ \mathcal{N} $ and $ \mathcal{T} $, and a set of edges $ \mathcal{E} \subset \mathcal{N} \times \mathcal{T} $ connecting them. These graphs naturally model systems where entities from one group are to be matched with entities from another, such as tasks to agents, products to consumers, or ads to users. A \emph{matching} in such a graph refers to a subset of edges with no shared endpoints — ensuring that each entity is involved in at most one pair. The core challenge lies in computing optimal matchings that respect resource constraints and maximize some utility or coverage. This problem has significant practical relevance, especially in operations research, where it is linked to the classical assignment problem~\cite{Grove}.

Recent real-world applications, particularly in online advertising, ride-hailing platforms, and real-time job allocation, have attracted significant interest in the online version of the matching problem~\cite{mehta_survey}. In this setting, nodes in $\mathcal{N}$ (e.g., advertisers or servers) are fixed and known in advance, while nodes in $\mathcal{T}$ (e.g., users or requests) arrive sequentially. When a new node $t \in \mathcal{T}$ arrives, the algorithm must decide on the spot whether to match it to an available node $n \in \mathcal{N}$, such that $(n, t) \in \mathcal{E}$, with the constraint that each node can be matched at most once. These decisions are irrevocable, making the problem both practically and theoretically challenging. The central goal is to design algorithms that construct a maximum matching—that is, one that covers as many nodes as possible. Typically, such algorithms are analyzed in one of two ways: by approximating the size of the resulting matching, or by evaluating their competitive ratio, defined as the worst-case ratio between the size of the algorithm’s matching and that of an optimal matching computed with full knowledge of the graph in advance.

Online matching has been explored through several theoretical lenses, most notably the adversarial and stochastic frameworks \cite{mehta_survey}. In the adversarial setting, the graph and arrival sequence are designed to be worst-case, providing robust but conservative performance guarantees. In contrast, stochastic models assume randomness in either the graph structure or the arrival process \cite{borrodin}, allowing for more realistic analyses and often stronger guarantees. Within this stochastic line of work, much attention has been given to Erdős–Rényi-type models \cite{mastin_jaillet}, where edges are included independently with a fixed probability. While such models are analytically tractable and provide insights, they fail to capture complex patterns like community structure, heterogeneity, and group-based interactions observed in real-world systems \cite{sbm_overview}. 

To address these limitations, the \emph{stochastic block model} (SBM) has emerged as a powerful alternative \cite{sbm_overview,abbe2018community}. SBM introduces latent classes (or communities) and allows edge probabilities to depend on class membership, thereby capturing structured heterogeneity—such as homophily (the tendency of similar nodes to connect more frequently)\cite{mcpherson2001birds}, core-periphery patterns (where a densely connected "core" group links to many others while "periphery" nodes have fewer connections) \cite{borgatti2000models}, or inter-group affinities (specific preferences or tendencies for nodes in one group to connect with nodes in another group)\cite{airoldi2008mixed}. This makes SBM particularly well-suited for modeling interactions in social networks, recommendation engines, and online marketplaces, where the likelihood of a match depends not just on individual attributes but also on group-level dynamics \cite{karrer2011stochastic}. In the context of online bipartite matching, this leads naturally to the \emph{bipartite stochastic block model}, where one partition consists of fixed agents and the other of arriving users, both belonging to latent classes that govern connection probabilities \cite{kim2020stochastic}.

We study online matching in the bipartite stochastic block model, focusing on the sparse regime, where the average degree of each node remains bounded as the system grows. This regime is particularly relevant in practice, as real-world platforms typically feature users or items that interact with only a small subset of the population. Sparsity not only reflects these empirical network structures but also introduces significant analytical challenges. Formally, we consider a bipartite graph $ \mathcal{G} = (\mathcal{N}, \mathcal{T}, \mathcal{E}) $, where $ \mathcal{N} $ is a fixed set of nodes and $ \mathcal{T} $ is a set of nodes that arrive sequentially. Each node $ n \in \mathcal{N} $ is independently assigned a class $ c(n) \in \mathcal{C} $ according to a distribution $ \mu(n) $, and each arriving node $ t \in \mathcal{T} $ is independently assigned a class $ d(t) \in \mathcal{D} $ with distribution $ \nu(t) $. Conditional on these class assignments, an edge between $ n $ and $ t $ is present independently with probability $ p_{n,t}$, which depends on their respective classes. A defining feature of this model, motivated by real-world constraints, is that the set of edges incident to each arriving node is not known in advance. Instead, when a node $ t \in \mathcal{T} $ arrives, the algorithm observes only then which edges exist between $ t $ and the nodes in $ \mathcal{N} $. This assumption captures a key operational reality in many systems, where information about potential interactions is revealed only upon arrival or activation of a new entity. For example, in online job platforms, the compatibility between a newly posted job (a node in $ \mathcal{T} $) and existing freelancers (nodes in $ \mathcal{N} $) becomes clear only when the job description is published. The system cannot precompute all possible matches due to computational constraints and the dynamic, user-driven nature of postings. In recommendation systems, user preferences are inferred from behavior observed at login, and only then can the platform determine which content is relevant — effectively modeling the appearance of edges at the moment of interaction. Given this online nature of information revelation, an algorithm must operate under uncertainty: it observes each vertex in $ \mathcal{T} $ sequentially and must decide, upon arrival, whether and how to match it to an available node in $ \mathcal{N} $, without knowledge of future arrivals or their connections. In this work, we introduce and analyze two natural algorithms designed for this setting:

\begin{itemize}
    \item \textbf{\myopic}: Upon the arrival of a node $ t \in \mathcal{T} $, the algorithm chooses a compatible class $ c^* \in \mathcal{C} $ and attempts to match $ t $ with an available node from this class.
    \item \textbf{\balance}: This algorithm selects the class with the highest probability of a successful match, considering both compatibility and current availability.
\end{itemize}
We focus on the $\myopic{}$ and $\balance{}$ algorithms because they are simple, practical, and reflect decision-making heuristics commonly used in real-world systems. Notably, while the $\myopic{}$ algorithm has been studied in stochastic models, the $\balance{}$ algorithm—despite its widespread use—has only been analyzed in adversarial settings. Its theoretical performance in structured stochastic environments, such as the bipartite stochastic block model, remains unexplored. These algorithms are appealing not only for their practical relevance but also for their interpretability and ease of implementation, which is particularly valuable in applications like online marketplaces, content recommendation, or allocation systems. We first analyze them under the assumption that the compatibility probabilities $p_{n,t}$ between user and item classes are known. This setting provides a tractable analytical framework, allowing us to rigorously characterize the fluid-limit behavior of both algorithms. While this fully informed setting offers valuable insights, it often does not reflect the realities faced by many real-world systems. Motivated by these practical constraints, we then turn to the more realistic case where the probabilities $ p_{n,t} $ are not known a priori. In many applications—such as recommendation engines or online platforms—the interaction propensities between user and item types must be inferred over time through observed outcomes. This naturally gives rise to a bandit setting, where the algorithm must estimate the unknown affinities $ p_{n,t} $ from binary feedback (indicating whether a match succeeded or failed), while simultaneously making irrevocable matching decisions. This introduces an exploration-exploitation trade-off that is absent in the known-parameter regime. To address this challenge, we propose and analyze a bandit version of the \balance{} algorithm, which learns class affinities dynamically while aiming to preserve strong matching performance.

Our two main contributions, correspond to the two settings considered:

\begin{itemize}
    \item \textbf{When $ p_{n,t} $ are known:} We provide a fluid-limit analysis of both $\myopic{}$ and $\balance{}$ algorithms in the sparse bipartite stochastic block model. Specifically,
    \begin{itemize}
        \item We prove that the matching size obtained by the $\myopic{}$ algorithm is, with high probability, close to the solution of a specific ordinary differential equation. Due to the complexity of solving this equation in closed form, we derive a tractable approximation and show that the resulting error remains small.
        \item For the $\balance{}$ algorithm, we prove that the matching size converges with high probability to a solution of a differential inclusion—a generalization of ODEs that captures the algorithm's discontinuous decision rules. To our knowledge, this is the first use of differential inclusions to analyze online matching problems.
        \item We extend this analysis to a generalized version of \balance{}, where the decision rule incorporates both connection probabilities and real-time availability, and show that its performance similarly converges to a well-defined differential inclusion.
    \end{itemize}

    \item \textbf{When $ p_{n,t} $ are unknown:} We study $\learbalance$ algorithm a bandit extension of the $\balance$ algorithm, where the compatibility probabilities $ p_{n,t} $ are unknown and must be learned over time. In this setting, the algorithm receives binary feedback for each matching attempt and must estimate the latent affinities between classes while making sequential, irrevocable decisions. We analyze the regret of this learning-based algorithm and prove that it is of order $\cO(T^{\frac{q+3}{4}})$ for $0<q<1$ using stochastic approximations and differential inclusions tools.
\end{itemize}
\subsection*{Related works}
Online bipartite matching has been extensively studied, particularly in adversarial and stochastic models (\cite{surveyhuang,mehta} for a survey). In the adversarial setting, the $\greedy$ algorithm guarantees a $1/2$ competitive ratio, improving to $1 - 1/e$ under random arrivals \cite{mehta_goel}. The $\ranking$ algorithm achieves the optimal $1 - 1/e$ bound in this setting and performs even better with random arrivals \cite{karp_vazirani, Devanur_ranking, mahdian_yan}. In contrast, stochastic models assume known distributions over vertex types, often in the i.i.d. setting. This allows improved performance, with algorithms reaching competitive ratios up to ~0.711 \cite{Manshadi_sto, jaillet_liu, Brubach2016Online, huang_shu}. However, the i.i.d. model overlooks graph structure, and in many practical or average-case settings, simple heuristics can match or outperform these algorithms \cite{brrodin}. This has motivated the study of stochastic input models that better reflect real-world graphs. Consequently, another line of work focuses on applying online algorithms to specific random graph families. A foundational example is online matching in Erdős–Rényi graphs, particularly in the sparse regime where each edge exists independently with probability $c/N$ \cite{mastin_jaillet, second_member_approx, dyer_frieze}. Even for simple strategies like $\greedy$, analysis in this setting is nontrivial and yields valuable insights. The configuration model further generalizes this approach by prescribing degree distributions for vertices \cite{noiry2021online, aamand2022optimal}. Additional generalizations of Erdős–Rényi have introduced dynamic elements—for instance, models where node degrees evolve over time to reflect changing environments or behaviors \cite{cherifaal}. The stochastic block model (SBM), a structured extension of Erdős–Rényi that captures community structure, has also been studied in the online setting. In the dense regime, \cite{sopranoloto2023online} analyze max-weight policies and show they can achieve asymptotically perfect matchings. In the general SBM, \cite{brandenberger} extend lower bounds from the Erdős–Rényi case, demonstrating that the 0.837 bound from \cite{mastin_jaillet} remains tight when communities have equal expected degrees. They propose several efficient heuristics for online matching in SBMs, which perform well empirically, but none are proven to achieve asymptotic optimality. However, all these works focus primarily on dense regimes or heuristics without theoretical guarantees. In particular, the sparse regime of SBM, which is highly relevant for many real-world applications where connections are scarce and structured, has received limited attention. From another perspective, bandit algorithms offer a general framework for decision-making under uncertainty, where limited feedback guides sequential choices. These models focus on balancing between exploration and exploitation and have found broad application in online learning and resource allocation \cite{lattimore,Slivkins}. While conceptually distinct, they share core challenges with online matching and offer complementary insights.
\section{Model}
\label{subsec:model}
We consider the online bipartite matching problem, where the nodes on one side arrive sequentially, with an additional graph structure given by a \emph{stochastic block model}. The latter is defined by a bipartite graph $ \mathcal{G} = (\mathcal{N}, \mathcal{T}, \mathcal{E}) $, where $ \mathcal{N} = [N] := \{1, \ldots, N\} $ is the set of ``offline'' nodes, $ \mathcal{T} = [T] $ is the set of ``online'' nodes, and $ \mathcal{E} \subset \mathcal{N} \times \mathcal{T} $ is the set of edges. This underlying graph is random, in the sense that each edge $ (n,t) \in \mathcal{N} \times \mathcal{T} $ belongs to $ \mathcal{E} $ independently with some probability. The block model assumes that each node belongs to a latent class: we denote by $ \mathcal{C} := [C] $ the set of classes on the offline side, and by $ \mathcal{D} := [D] $ the set of classes on the online side. Each offline node $ n \in \mathcal{N} $ is assigned a class $ c(n) \in \mathcal{C} $, and each online node $ t \in \mathcal{T} $ is assigned a class $ d(t) \in \mathcal{D} $. These assignments are drawn independently: nodes on the offline side are sampled from a distribution $ \mu $ over $ \mathcal{C} $, and nodes on the online side are sampled from a distribution $ \nu $ over $ \mathcal{D} $. Given the class labels, the edge $ (n, t) $ appears in $ \mathcal{E} $ with probability $ p_{n,t} = p\big(c(n), d(t)\big) \in [0,1] $, where $ p = (p(c,d))_{c,d \in \mathcal{C} \times \mathcal{D}} $ is a class-to-class affinity matrix. In this work, we focus on the \emph{sparse regime}, where the underlying graph remains sparse as the number of offline nodes $ N $ grows. More precisely, we assume the existence of a non-negative matrix $ a = (a_{c,d}) \in \mathbb{R}_+^{C \times D} $ such that $p(c,d) = \frac{a_{c,d}}{N}$. Moreover, we assume that \( a_{c,d} \leq a \) for all \( c \in \mathcal{C} \), \( d \in \mathcal{D} \), where \( a \in (0, N) \) is a fixed constant.
This choice ensures that each offline node has a bounded expected degree, even as $ N \to \infty $, which reflects realistic constraints in large-scale platforms where individual users or items interact with only a limited number of others.

As mentioned in the introduction, in the online matching problem, an algorithm $\mathsf{ALG}$ observes sequentially the vertices in $\mathcal{T}$ and constructs on the fly a matching (i.e., a subset of edges such that any vertex belongs to at most one of them) irrevocably: after seeing the vertex $t \in \mathcal{T}$, it can decide to add irrevocably an edge $(n,t) \in \mathcal{E}$ to the current matching if $t$ does not belong to an edge of the matching yet.

We shall now introduce some notations. We denote by $b_c$  the proportion of nodes of $\mathcal{N}$ of class $c \in \mathcal{C}$. We also assume that there exists some scaling factor $\alpha >0$ such that $T=\alpha N$ (as we will consider asymptotic results when $N$ is large). We will also define the Boolean variable $m_{n}(t)$ equal to 1 if, and only if, the vertex $u$ has been included in the matching by $\mathsf{ALG}$ before the vertex $t$ arrives (otherwise $m_{n}(t)=0$). Additionally, we denote by $\mathcal{N}_c:=\{n \in \mathcal{N}, c(n)=c\}$ the set of nodes of class $c \in \mathcal{C}$, by $\mathcal{M}_c(t)=\{n \in \mathcal{N}_c, m_{n}(t)=1\}$ the set of vertices of class $c$ that are already matched before seeing vertex $t \in \mathcal{T}$ (we denote by $M_c(t)$ its cardinality) and by $M(t):=\sum_{c\in\mathcal{C}}M_c(t)$ the size of the matching constructed so far. We also denote by $\mathcal{F}_c(t)=\{u \in \mathcal{N}_c\backslash\mathcal{M}_c(t), (n,t) \in \mathcal{E}\}$ the set of vertices of class $c$ that are not matched so far (thus free), but such that $(n,t)$ belongs to $\mathcal{E}$ (they are the ``available neighbors'' of $t$ of class $c$), and by $C(t)= \{c\in [C] | \mathcal{F}_c(t)\neq \emptyset\}$ the set of classes that have available nodes at time $t\in [T]$.  Finally, we shall denote by $e_i$  the  $i$-th basis vector of $\lR^{C}$.

\section{Known compatibility probabilities}
\subsection{A warm-up: Myopic algorithm}

In this section, we present a simple yet foundational algorithm, $\myopic$, which serves as a baseline for more sophisticated algorithms presented later. This algorithm is designed to make fast, greedy decisions for matching, without attempting to look ahead or anticipate future availability. Specifically, when a new vertex $t \in \mathcal{T}$ arrives (e.g., a request or a user), the policy \emph{selects a class} $c_t \in \mathcal{C}$ according to a fixed probability distribution, and then \emph{attempts to match $t$} to an available node within that class. The selection is made without verifying beforehand whether any nodes in the chosen class are actually available at time $t$. As a result, $\myopic$ is computationally simple and immediate in its decisions, but it may sometimes fail to make a match due to resource unavailability. Despite its reactive nature, $\myopic$ is carefully designed to respect class-specific budget constraints and to maximize the expected long-term success rate of matches. The key component of the algorithm is the computation of a probability matrix $Q^*(c,d)$, which solves the following optimization problem:
\begin{align*}
 &Q^* \in \arg\max_Q \sum_{c,d} Q(c,d)p(c,d), \  \\ \text{s.t}~~\sum_d Q(c,d)\nu(d)=b_c, & ~\forall c \in \mathcal{C} ~~~~\text{and}~~~~~
  \sum_c Q(c,d)\nu(d)=\nu(d), \forall d \in \mathcal{D}.
\end{align*}
In particular, $Q^*$ represents the optimal transport plan that maps the distribution $\nu$ to the budget vector $b$, minimizing the transport cost with respect to $-p(c,d)$. This matrix can be computed efficiently using the Hungarian algorithm, with a computational complexity of $\mathcal{O}(CD(C+D))$. Notably, the special case where there is only one class (i.e., $C = D = 1$) reduces to the setting studied in \cite{mastin}.

\begin{algorithm}[H]
\caption{$\myopic$ policy}
\label{alg:greedy_algo}
\SetAlgoNlRelativeSize{0}
\DontPrintSemicolon
\KwOut{Updated matching $M(t)$}
Compute the optimal transport plan $Q^*$\;
\For{$t \in [T]$}{
Choose $c_t \in \mathcal{C}$ at random with probability $Q^*(c_t,d_t)/\nu(d_t)$

\If{$\mathcal{F}_{c_t}(t) = \emptyset$}{
    $M(t) = M(t-1)$\;
}
\Else{
    $M(t) = M(t-1) \cup \{(n_t, t)\}$ for $n_t \sim \text{unif}(\mathcal{F}_{c_t}(t))$
    }}
\end{algorithm}
In order to understand the behavior of the $\myopic$ policy in large-scale matching markets, we consider a deterministic approximation via an ordinary differential equation (ODE). The guiding intuition is that, as the number of agents grows (nodes in the graph), stochastic variability in the system averages out, and the evolution of the system can be captured by a smooth deterministic trajectory. However, the ODE associated with the dynamics of the $\myopic$ policy \Cref{eq:edo_greedy} is nonlinear and does not generally admit a closed-form solution, especially due to the complex structure of the underlying  graph encoded in the parameters $a_{c,d}$ and $Q^*(c,d)$. This makes direct analysis challenging. In certain structured settings, however, the ODE becomes more tractable. For instance, when $a_{c,d} = c$ for all $(c,d)$ which corresponds to the Erdős–Rényi model, the ODE simplifies and admits a closed-form solution, as shown in \cite{mastin_jaillet}. This illustrates how the structure of the graph can significantly affect the tractability of the analysis and motivates the need for approximation techniques in the more general, heterogeneous case.

To address this, we introduce a simplified surrogate ODE, which approximates the behavior of the original system while being analytically tractable. This auxiliary equation has a closed-form solution $\tilde{y}_c(t)$, allowing us to extract structural insights such as convergence rates and equilibrium behavior. Crucially, we rigorously control the error between the true ODE solution $y_c(t)$ and the approximate solution $\tilde{y}_c(t)$ by bounding it with an explicit error term $e_c(t)$.

The next theorem formalizes this two-step approximation strategy. It shows that:
\begin{itemize}
    \item The normalized matching size produced by the $\myopic$ policy closely follows the ODE solution $y_c(t)$ with high probability.

\item The ODE solution is well-approximated by the simpler function $\tilde{y}_c(t)$ with a small and explicitly controlled error.
\end{itemize}

\begin{restatable}{theorem}{approxsolution}\label{theorem:approxsolution}
Let $y_c : [0, \alpha] \to \mathbb{R}$ be the solution of the following ODE
\begin{align}
    \begin{cases}
    \dot{y}_c(s) &= \sum_{d=1}^{D} \left(1 - e^{-a_{c,d}(b_c - y_c(s))} \right)Q^*(c,d)\\
    y_c(0) &= 0
    \end{cases}
    \label{eq:edo_greedy}
\end{align}

Then, for each class $c \in \cC$, the matching size $M_c(t)$ produced by $\myopic$  satisfies, for all $t \in [T]$
\begin{align}
   \left| \frac{M_c(t)}{N} - y_c(t/N) \right| \leq  \frac{3 L_ce^{\alpha L_c}}{N^{1/3}}, \ \text{where}\ L_c = \sum_{d=1}^{D} a_{c,d} Q^{*}(c,d)
\end{align}
with probability at least $1-2Ce^{- N^{1/3}L_c^2/ 8\alpha}$. Moreover, for $ c \in [C] $, $ y_c(t)= \tilde{y}_c(t)-e_{c}(t)$, where $ e_c(0) = 0 $,  and $ \tilde{y}_c(t)= b_c-b_c\exp\left( -t L_c \right)$, and $e_c$ satisfies, 
$$e_c(t)\leq  \frac{J_c}{L_c}(1-e^{-L_c t}) $$
    Where  $J_c= \frac{b_c^2}{2}\sum_{d=1}^{D}a_{c,d}^2 Q^{*}(c,d)$.
\end{restatable}
\subsection{$\balance$ }
The $\balance$ algorithm builds on the limitations of the $\myopic$ policy by trying to take into account node availability. When a node of class $c(t)$ arrives, the $\myopic$ policy selects a compatible class $c$ purely based on potential compatibility between the classes—without considering whether any unmatched nodes from that class are actually available at that moment. In contrast, $\balance$ takes a more informed approach: it chooses the class $c$ that maximizes the probability that at least one unmatched node is available and connected to the arriving node. This probability is estimated by the expression $1 - \left(1 - \frac{a_{c,j}}{N}\right)^{N b_c - M_c(t)}$, which captures the likelihood of an edge existing under a stochastic block model, where edge probabilities between classes are governed by parameters $a_{c,j}$. 
\begin{algorithm}[H]
\caption{$\balance$}
\label{alg:balance_algo}
\SetAlgoNlRelativeSize{0}
\DontPrintSemicolon
\KwOut{Updated matching $M(t)$}
\For{$t \in [T]$}{

Choose $c_t= \argmax_{c\in [1,C]} \sum_{j=1}^{D}(1-(1-\frac{a_{c,j}}{N})^{Nb_c-M_c(t)})\nu(j)$ \; 

\uIf{${\cal F}_{c_t}(t) = \emptyset$}{
    $M(t) = M(t-1)$\;
}
\Else{
        $M(t) = M(t-1) \cup \{(n_t, t)\}$ for $n_t \sim \text{unif}(\mathcal{F}_{c_t}(t))$
}
}
\end{algorithm}
A key distinction between the $\myopic$ and $\balance$ policies lies in how they respond to the evolving state of the system—and this difference has important implications for their large-scale behavior. The $\myopic$ approach selects a class based solely on static compatibility between node types, leading to smooth dynamics. As the system scales, the randomness introduced by node arrivals averages out, and the evolution of the system can be accurately described by an ordinary differential equation (ODE). This continuous, deterministic approximation leverages the fact that the $\myopic$ policy’s decision rule is smooth and Lipschitz-continuous. In contrast, the $\balance$ policy takes a more strategic decision by selecting the class that maximizes the actual match probability. This introduces discontinuities into the process—small variations in the system state can lead to abrupt changes in the selected class. As a result, the system no longer evolves smoothly, and the assumptions required for ODE convergence break down.

To address this, we turn to the framework of differential inclusions, a generalization of ODEs designed to handle such discontinuous dynamics. Rather than prescribing a single trajectory, a differential inclusion allows the system's evolution to follow a set of possible directions at each point, capturing the non-smooth transitions in behavior driven by abrupt changes in decision rules. The following theorem formalizes this connection by proving that, with high probability, the normalized matching sizes produced by $\balance$ converge to a solution of a differential inclusion. For an introduction to differential inclusions and their relevance in this context, see \Cref{sec:differential_inclusion}.
\begin{restatable}{theorem}{inclusiondiffsol}\label{theorem:inclusiondiffsol}
Let $m$ be the unique solution of the differential inclusion $$\dot{m}\in F(m):= \text{conv}\left\{f_{c,b_{c}}(m_c)e_c \ ; \   c\in\argmax_{k\in [C]}f_{k,b_k}(m_k) \right\},$$ which is the convex hull of the mappings $$ f_{c,b_{c}}(x)= \sum_{d=1}^{D} (1-e^{-a_{c,d}(b_c-x)})\nu(d).$$

Then the matching  built by $\balance$ satisfies for all $t\in [T]$ and $c\in \mathcal{C}$,  with probability at least $1-\frac{b\alpha}{N\epsilon^2}$, 
 \begin{align}
        \left|\frac{M_c(t)}{N} -m_c(t/N)\right|\leq \min(\alpha, e^{L\alpha}/\sqrt{2L})\sqrt{A_{\alpha,c}/N+ \delta_c B_{\alpha,c}+\epsilon C_{\alpha,c}}, 
    \end{align}
   The different constants in are defined by $L= \max_{c\in [C]} \sum_{d=1}^{D} a_{c,d}\nu(d)$, $\delta_c= \frac{1}{N}\sum_{d=1}^{D}\frac{a_{c,d}}{e}\nu(d)$, $K_{\alpha}=(c\alpha+\epsilon)e^{c\alpha}/c$, $\epsilon$ as defined in \Cref{lemma:18gast} and $c$ in \Cref{lemma:assum_gast1}, $U_c= \sum_{d=1}^{D}(1-e^{-a_{c,d}b_c})\nu(d)$, $A_{\alpha,c}= U_c(U_c^2+\frac{14U_c}{3}+2K_{\alpha}), B_{\alpha,c}= 2U_c^2+4L\delta_c+12K_{\alpha}, C_{\alpha,c}= 2U_c^2+4L\epsilon +8K_{\alpha}$.
\end{restatable}

As \balance algorithm always picks the class with the highest probability of connection, which decreases in case of match, it tends to progressively equalize these probabilities across the classes in $\cC$. Once two classes have their probabilities (almost) equalized, these stay equal over time by decreasing at the same rate. This generates phases, in which the $k$ "first" classes (with highest initial probability) have their probabilities equalized, decreasing at the same rate, up until reaching the probability of the $(k+1)^{\rm th}$ class. This phasing allows to obtain an explicit formula for $m$. The remainder of the section introduces the exact formula and some intuitions of how it is built, while the verification proof is deferred to \Cref{subsec:proof_solinclusiondiff}. 

In the large-$N$ limit, the marginal probability that the next online node is connected to at least one node in class $c \in \mathcal{C}$, given a vector $\boldsymbol{\beta} \in \mathbb{R}_+^C$ representing the proportions of available nodes in each class, is given by: $f_{c, \beta_c}(z) = \sum_{d\in\cD} \left(1-e^{-a_{c,d}(\beta_c-z)}\right) \nu(d) \,$.
W.l.o.g., we assume for the analysis that the elements of $\cC$ are ordered by decreasing marginal probability of receiving at least one edge at the initial time step -- i.e.  $f_{1, b_1}(0) \geq f_{2, b_2}(0) \geq \dots \geq f_{C, b_C}(0)$. Additionally, as $a_{c,d} > 0$ 
, $f_{c, \beta_c}$ is strictly decreasing and thus invertible, with $f_{c, \beta_c}^{-1}$ also strictly decreasing. 

\paragraph{During phase $k$,} the $C-k$ last classes are not selected at all, while the $k$ first classes have their probabilities equalized, decreasing at the same rate. Given the budgets $\bm{\beta} \in \mathbb{R}^C_+$ at the beginning of the phase, the number of  nodes matched during the phase, in each of the $k$ classes, at time $t$, ends up evolving following $\mu_{k,\bm{\beta}}(t)$ defined as the solution of the following separable ODE
\begin{align*}
    \begin{cases}
        \frac{{\rm d}\mu_{k,\bm{\beta}}}{F_{k,\bm{\beta}}(\mu_{k,\bm{\beta}})} = {\rm d}t\\
        \mu_{k,\bm{\beta}}(0) = 0
    \end{cases}
    ~~~~~\text{where}~~~~~
    F_{k, \bm{\beta}} = \left(\sum_{c=1}^k f_{c, \beta_c}^{-1}\right)^{-1}\,.
\end{align*}
Note that, as $\frac{{\rm d}\mu_{k,\bm{\beta}}}{{\rm d}t} > 0$, $\mu_{k,\bm{\beta}}$ is strictly increasing, positive and invertible.

\paragraph{To assemble the phases,} and provide the full expression of $m$, two sequences need to be defined. First, the sequence of time-steps $(t_k)_{k \in [C]}$ defines the phases. Then $\left(\bm{\beta}^{(k)}\right)_{k\in[C]}$ describes the proportion of available nodes per class at the start of each phase. They are defined as follows:

\begin{align}
    \forall c, k\in[C], ~\bm{\beta}_c^{(k)} &= \begin{cases}
        b_c & \text{ if } k \leq c\\
        \bm{\beta}_c^{(k-1)} - f_{c, \bm{\beta}^{(k-1)}}^{-1}\left(f_{k, \bm{\beta}^{(1)}}(0)\right) & \text{ otherwise}
    \end{cases}
    \\
    \forall k\in[C+1], ~ t_k &= \begin{cases}
        0 & \text{ if }k = 1\\
        \min\left(T, t_{k-1} + \mu_{k-1,\bm{\beta}^{(k-1)}}^{-1}\left(F_{k-1,\bm{\beta}^{(k-1)}}^{-1}\left(f_{k,\bm{\beta}^{(1)}}(0)\right)\right)\right) & \text{otherwise }
    \end{cases}
\end{align}

\begin{restatable}{theorem}{explicitsolution}\label{theorem:explicitsolution}
For any $c\in\cC$, the function
\begin{align}
    m^*_c : t \mapsto (b_c - \bm{\beta}^{(k_t)}_c) + \left(f_{c, \bm{\beta}^{(k_t)}}^{-1}(F_{k_t, \bm{\beta}^{(k_t)}}(\mu_{k_t,\bm{\beta}^{(k_t)}} (t - t_{k_t})))\right)_+
\end{align}
where $k_t = \max \{k\in[C] : t > t_k\}$,  is thel solution of the differential inclusion defined in \Cref{theorem:inclusiondiffsol}.
\end{restatable}
\Cref{fig:illustration_fluid_limit_and_regret} illustrates the accuracy of $m^*_c(t/N)$ to estimate $\frac{M_c(t)}{N}$ in the sparse regime.

\subsection{$\realbalance$}
While the $\balance$ policy improves upon the purely compatibility-driven $\myopic$ approach by considering the probability of successful matches, it still relies on an expected availability model—it selects the class that likely has unmatched nodes, without verifying their presence. This can lead to wasted opportunities when the chosen class ends up being empty. To address this limitation, we introduce the $\realbalance$ policy, which takes an additional, more grounded step. Instead of only maximizing the expected match probability, $\realbalance$ ensures that the selected class actually contains at least one available node at the time of decision. This additional check prevents the algorithm from targeting unavailable options, making it more efficient. We prove that the matching size under $\realbalance$ converges with high probability to a differential inclusion. This differential inclusion differs from the one in \Cref{theorem:inclusiondiffsol}, as it explicitly accounts for real-time availability constraints. Full details and results are provided in \Cref{subsec:real_balance}.
\section{Unknown compatibility probabilities}
In many real-world applications, the underlying parameters of the graph—such as the connection probabilities between node classes or the distribution of arriving node types—are not known \emph{a priori}. These parameters may be shaped by latent variables, evolve over time, or be inferred only through noisy and partial observations. Consequently, algorithms must operate under uncertainty and learn these parameters dynamically in order to make effective matching decisions. In this section, we study the setting in which the connection probabilities $ a_{c,d} $ are unknown and must be estimated online. This transforms the problem into a bandit setting, where each class $ c \in \mathcal{C} $ can be viewed as an arm, and upon choosing class $ c_t $ at time $ t $, a Bernoulli reward is observed—indicating whether a successful match occurred between the arriving node and an available node in class $ c_t $.

\subsection{$\learbalance$}
Algorithm~\ref{alg:learbalance_concentration} presents the $\learbalance$ policy, which combines a fixed-duration Explore-Then-Commit (ETC) strategy with the $\balance$ rule for class selection. The algorithm proceeds in two phases. During the exploration phase, which lasts for a fixed number of rounds \( T_{\text{explore}} \), arriving nodes are matched by selecting classes uniformly at random. This allows the algorithm to collect data on the outcomes of triplets \((c, d, m)\), where \(c \in \mathcal{C}\), \(d \in \mathcal{D}\), and \(m\) is the current matching size of the class \(c\). We let \(T_{c,d,m}\) denote the number of times the triplet \((c, d, m)\) has been observed. After this phase, the algorithm enters the commitment phase and uses the collected data to estimate the match failure probabilities. Specifically, it estimates \(D_{c,d}(m) = \left(1 - \frac{a_{c,d}}{N} \right)^{N b_c - m}\) using an estimator \(\hat{D}_{c,d}(m)\), whose form and concentration properties are given in \Cref{theorem:concentrationpowers}. These estimates are then incorporated into the $\balance$ rule to select the class \(c\) that maximizes the expected success probability, weighted by the type distribution \(\nu(d)\).
 
\begin{algorithm}[ht]
\caption{$\learbalance$ }
\label{alg:learbalance_concentration}
\DontPrintSemicolon
\KwIn{$T$, $\nu(d)$, $N$, $T_{\text{explore}}$, confidence level $\delta$}
\KwOut{Matching $M(t)$}

\textbf{Init:} $\mathcal{M}(0) = \emptyset$, $M_c(0) = 0$, $T_{c,d,m}(0) = 0$, $\hat{D}_{c,d}(m) = 1$\;

\For{$t = 1$ \KwTo $T$}{
    Node $t$ of type $d(t)$ arrives\;

    \uIf{$t \leq T_{\text{explore}}$}{
        Select $c_t$ uniformly at random\;
    }
    \Else{
        \ForEach{$c \in [C]$}{
            Let $m = M_c(t-1)$, $s(m) = \frac{b_c}{b_c - m/N}$\;
            Compute $\hat{D}_{c,d}(m)$ 
        }
        Choose $c_t = \arg\max_c \sum_d 1 - \hat{D}_{c,d}(M_c(t{-}1)) \nu(d)$\;
    }

    \uIf{ $ \mathcal{F}_{c_t}(t) \neq \emptyset$ }{
        Match $(n_t,t)$ for $n_t \sim \text{unif}(\mathcal{F}_{c_t}(t))$, update $\mathcal{M}(t)$, $M_{c_t}(t)$, and $Y_{c_t,d(t)}(t) = 1$\;
    }
    \Else{
        No match, $Y_{c_t,d(t)}(t) = 0$\;
    }

    Update $T_{c_t,d(t),M_{c_t}(t)} \gets T_{c_t,d(t),M_{c_t}(t)} + 1$\;
}
\end{algorithm}
\subsection{Regret}
For each class $ i \in \mathcal{C} $, let $ M_i(t) $ and $ \hat{M}_i(t) $ denote the number of matches made by the $\balance$ and $\learbalance$ algorithms, respectively, up to time $ t $. We define the regret of $\learbalance$ as the total difference in matching performance across all classes when compared to $\balance$, which has full knowledge of the match probabilities. In other words, the regret quantifies the cumulative loss incurred by $\learbalance$ due to not knowing the match probabilities in advance. The following result shows that this regret grows at most on the order of $ \mathcal{O}(N^{(q+3)/4}) $ for some $0<q<1$.
\begin{restatable}{theorem}{regretetcbalance}\label{theorem:regretetcbalance}
Let $ R(T) = \sum_{i \in \mathcal{C}} M_i(T) - \hat{M}_i(T) $ denote the regret of $\learbalance$. Suppose the exploration phase lasts for $ T_{\text{explore}} = T^{\frac{q+3}{4}} $, for some $ 0 < q < 1 $. Then the regret satisfies
\[
R(T) = \mathcal{O}(T^{\frac{q+3}{4}}).
\]
\end{restatable}

\begin{figure}
    \centering
    \includegraphics[width=0.5\linewidth]{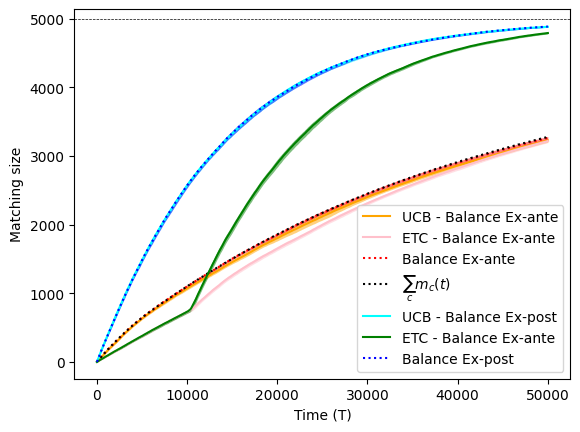}
    \caption{Illustration of the matching size for the different methods. $T=50000$, $N=5000$, $C=5$, $D=6$ and simulations are averaged over $20$ trajectories. {\bf 1.} As expected the \emph{ex-post} versions of the algorithms perform better than their \emph{ex-ante} counterparts.  {\bf 2.} The fluid limit $m^*$ (dark) is an accurate estimator of the actual empirical trajectory of \balance (red). {\bf 3.} A UCB version of balance can be built using the confidence set from \Cref{theorem:concentrationpowers} and performs empirically better than ETC. Yet, its analysis remains an open question.}
    \label{fig:illustration_fluid_limit_and_regret}
\end{figure}

\section{Conclusion}
\label{sec:conclusion}
We studied online bipartite matching within the Stochastic Block Model (SBM), capturing structured heterogeneity in real-world networks through class-dependent connection probabilities. We analyzed two main algorithms under known probabilities: the Myopic policy, which is simple but limited by ignoring availability, and the $\balance$ algorithm, which accounts for compatibility and capacity, and is shown to converge to a differential inclusion. When the probabilities are unknown, we introduced $\learbalance$, a bandit-based extension that learns affinities over time and achieves sublinear regret. Simulations confirm that UCB-based methods outperform $\learbalance$, while $\balance$ under known probabilities closely matches actual outcomes, validating its effectiveness. A promising future research direction is to analyze UCB in this setting using differential inclusion tools, to better understand its asymptotic behavior and theoretical guarantees.

\bibliographystyle{abbrvnat}
\bibliography{biblio}

\begin{thebibliography}{45}
\providecommand{\natexlab}[1]{#1}
\providecommand{\url}[1]{\texttt{#1}}
\expandafter\ifx\csname urlstyle\endcsname\relax
  \providecommand{\doi}[1]{doi: #1}\else
  \providecommand{\doi}{doi: \begingroup \urlstyle{rm}\Url}\fi

\bibitem[Aamand et~al.(2022)Aamand, Chen, and Indyk]{aamand2022optimal}
A.~Aamand, J.~Y. Chen, and P.~Indyk.
\newblock (optimal) online bipartite matching with degree information.
\newblock In A.~H. Oh, A.~Agarwal, D.~Belgrave, and K.~Cho, editors, \emph{Advances in Neural Information Processing Systems}, 2022.

\bibitem[Abbe(2018)]{abbe2018community}
E.~Abbe.
\newblock Community detection and stochastic block models: Recent developments.
\newblock \emph{Journal of Machine Learning Research}, 18\penalty0 (177):\penalty0 1--86, 2018.

\bibitem[Airoldi et~al.(2008)Airoldi, Blei, Fienberg, and Xing]{airoldi2008mixed}
E.~M. Airoldi, D.~M. Blei, S.~E. Fienberg, and E.~P. Xing.
\newblock Mixed membership stochastic blockmodels.
\newblock In \emph{Advances in Neural Information Processing Systems (NeurIPS)}, 2008.

\bibitem[Aubin and Cellina(1984)]{aubin1984}
J.-P. Aubin and A.~Cellina.
\newblock \emph{Differential Inclusions: Set-Valued Maps and Viability Theory}, volume 264 of \emph{Grundlehren der mathematischen Wissenschaften}.
\newblock Springer-Verlag, Berlin, 1984.

\bibitem[Bollob{\'a}s(2001)]{bollobas2001random}
B.~Bollob{\'a}s.
\newblock \emph{Random Graphs}.
\newblock Cambridge University Press, 2001.

\bibitem[Borgatti and Everett(2000)]{borgatti2000models}
S.~P. Borgatti and M.~G. Everett.
\newblock Models of core/periphery structures.
\newblock \emph{Social Networks}, 21\penalty0 (4):\penalty0 375--395, 2000.

\bibitem[Borodin et~al.(2018)Borodin, Karavasilis, and Pankratov]{second_member_approx}
A.~Borodin, C.~Karavasilis, and D.~Pankratov.
\newblock Greedy bipartite matching in random type poisson arrival model, 2018.

\bibitem[Borodin et~al.(2020{\natexlab{a}})Borodin, Karavasilis, and Pankratov]{brrodin}
A.~Borodin, C.~Karavasilis, and D.~Pankratov.
\newblock An experimental study of algorithms for online bipartite matching.
\newblock \emph{ACM J. Exp. Algorithmics}, 25, mar 2020{\natexlab{a}}.
\newblock ISSN 1084-6654.
\newblock \doi{10.1145/3379552}.

\bibitem[Borodin et~al.(2020{\natexlab{b}})Borodin, MacRury, and Rakheja]{borrodin}
A.~Borodin, C.~MacRury, and A.~Rakheja.
\newblock Bipartite stochastic matching: Online, random order, and i.i.d. models.
\newblock 04 2020{\natexlab{b}}.
\newblock \doi{10.48550/arXiv.2004.14304}.

\bibitem[Brandenberger et~al.(2024)Brandenberger, Chin, Sheffield, and Shyamal]{brandenberger}
A.~Brandenberger, B.~Chin, N.~S. Sheffield, and D.~Shyamal.
\newblock {Matching Algorithms in the Sparse Stochastic Block Model}.
\newblock In C.~Mailler and S.~Wild, editors, \emph{35th International Conference on Probabilistic, Combinatorial and Asymptotic Methods for the Analysis of Algorithms (AofA 2024)}, volume 302 of \emph{Leibniz International Proceedings in Informatics (LIPIcs)}, pages 16:1--16:21, Dagstuhl, Germany, 2024. Schloss Dagstuhl -- Leibniz-Zentrum f{\"u}r Informatik.
\newblock ISBN 978-3-95977-329-4.
\newblock \doi{10.4230/LIPIcs.AofA.2024.16}.

\bibitem[Brubach et~al.(2016)Brubach, Sankararaman, Srinivasan, and Xu]{Brubach2016Online}
B.~Brubach, K.~A. Sankararaman, A.~Srinivasan, and P.~Xu.
\newblock Online stochastic matching: New algorithms and bounds.
\newblock \emph{Algorithmica}, 82:\penalty0 2737 -- 2783, 2016.

\bibitem[Cherifa et~al.(2024)Cherifa, Calauzenes, and Perchet]{cherifaal}
M.~Cherifa, C.~Calauzenes, and V.~Perchet.
\newblock Dynamic online matching with budget refills.
\newblock 05 2024.
\newblock \doi{10.48550/arXiv.2405.09920}.

\bibitem[Devanur et~al.(2013)Devanur, Jain, and Kleinberg]{Devanur_ranking}
N.~R. Devanur, K.~Jain, and R.~D. Kleinberg.
\newblock Randomized primal-dual analysis of ranking for online bipartite matching.
\newblock In \emph{Proceedings of the Twenty-Fourth Annual ACM-SIAM Symposium on Discrete Algorithms}, SODA '13, page 101–107, USA, 2013. Society for Industrial and Applied Mathematics.
\newblock ISBN 9781611972511.

\bibitem[Dyer et~al.(1993)Dyer, Frieze, and Pittel]{dyer_frieze}
M.~Dyer, A.~Frieze, and B.~Pittel.
\newblock The average performance of the greedy matching algorithm.
\newblock \emph{The Annals of Applied Probability}, 3\penalty0 (2):\penalty0 526--552, 1993.

\bibitem[Filippov(1988)]{filippov1988}
A.~F. Filippov.
\newblock \emph{Differential Equations with Discontinuous Right-Hand Sides}, volume~18 of \emph{Mathematics and Its Applications (Soviet Series)}.
\newblock Kluwer Academic Publishers, Dordrecht, 1988.
\newblock Translated from the Russian.

\bibitem[Gast and Gaujal(2011)]{gast:inria-00491859}
N.~Gast and B.~Gaujal.
\newblock {Markov chains with discontinuous drifts have differential inclusions limits. Application to stochastic stability and mean field approximation.}
\newblock Research Report RR-7315, Apr. 2011.

\bibitem[Godsil(1981)]{Godsil1981MatchingsAW}
C.~D. Godsil.
\newblock Matchings and walks in graphs.
\newblock \emph{J. Graph Theory}, 5:\penalty0 285--297, 1981.

\bibitem[Goel and Mehta(2008)]{mehta_goel}
G.~Goel and A.~Mehta.
\newblock Online budgeted matching in random input models with applications to adwords.
\newblock In \emph{Proceedings of the Nineteenth Annual ACM-SIAM Symposium on Discrete Algorithms}, SODA '08, page 982–991, USA, 2008. Society for Industrial and Applied Mathematics.

\bibitem[Grove et~al.(1995)Grove, Kao, Krishnan, and Vitter]{Grove}
E.~F. Grove, M.-Y. Kao, P.~Krishnan, and J.~S. Vitter.
\newblock Online perfect matching and mobile computing.
\newblock In \emph{Algorithms and Data Structures: 4th International Workshop, WADS'95 Kingston, Canada, August 16--18, 1995 Proceedings 4}, pages 194--205. Springer, 1995.

\bibitem[Holland et~al.(1983)Holland, Laskey, and Leinhardt]{sbm_overview}
P.~W. Holland, K.~B. Laskey, and S.~Leinhardt.
\newblock Stochastic blockmodels: First steps.
\newblock \emph{Social Networks}, 5\penalty0 (2):\penalty0 109--137, 1983.

\bibitem[Huang et~al.(2022)Huang, Shu, and Yan]{huang_shu}
Z.~Huang, X.~Shu, and S.~Yan.
\newblock The power of multiple choices in online stochastic matching.
\newblock In \emph{Proceedings of the 54th Annual ACM SIGACT Symposium on Theory of Computing}, STOC 2022, page 91–103, New York, NY, USA, 2022. Association for Computing Machinery.
\newblock ISBN 9781450392648.
\newblock \doi{10.1145/3519935.3520046}.

\bibitem[Huang et~al.(2024)Huang, Tang, and Wajc]{surveyhuang}
Z.~Huang, Z.~G. Tang, and D.~Wajc.
\newblock Online matching: A brief survey.
\newblock \emph{SIGecom Exch.}, 22\penalty0 (1):\penalty0 135–158, Oct. 2024.
\newblock \doi{10.1145/3699824.3699837}.

\bibitem[Jaillet and Lu(2014)]{jaillet_liu}
P.~Jaillet and X.~Lu.
\newblock Online stochastic matching: New algorithms with better bounds.
\newblock \emph{Mathematics of Operations Research}, 39\penalty0 (3):\penalty0 624--646, 2014.
\newblock ISSN 0364765X, 15265471.

\bibitem[Janson et~al.(2011)Janson, Łuczak, and Ruciński]{janson2011random}
S.~Janson, T.~Łuczak, and A.~Ruciński.
\newblock \emph{Random Graphs}.
\newblock Wiley Series in Discrete Mathematics and Optimization. Wiley, 2011.

\bibitem[Karp et~al.(1990)Karp, Vazirani, and Vazirani]{karp_vazirani}
R.~M. Karp, U.~V. Vazirani, and V.~V. Vazirani.
\newblock An optimal algorithm for on-line bipartite matching.
\newblock In \emph{Proceedings of the Twenty-Second Annual ACM Symposium on Theory of Computing}, STOC '90, page 352–358, New York, NY, USA, 1990. Association for Computing Machinery.
\newblock ISBN 0897913612.

\bibitem[Karrer and Newman(2011)]{karrer2011stochastic}
B.~Karrer and M.~E.~J. Newman.
\newblock Stochastic blockmodels and community structure in networks.
\newblock \emph{Physical Review E}, 83\penalty0 (1):\penalty0 016107, 2011.

\bibitem[Khalil(2002)]{Khalil:1173048}
H.~K. Khalil.
\newblock \emph{{Nonlinear systems; 3rd ed.}}
\newblock Prentice-Hall, Upper Saddle River, NJ, 2002.
\newblock The book can be consulted by contacting: PH-AID: Wallet, Lionel.

\bibitem[Kim and Oh(2020)]{kim2020stochastic}
M.~Kim and S.~Oh.
\newblock Stochastic blockmodel with cluster-dependent connection probabilities for modeling bipartite networks.
\newblock \emph{Journal of Machine Learning Research}, 21\penalty0 (234):\penalty0 1--43, 2020.

\bibitem[Kunze(2000)]{kunze2000}
M.~Kunze.
\newblock \emph{Non-Smooth Dynamical Systems}, volume 1744 of \emph{Lecture Notes in Mathematics}.
\newblock Springer, Berlin, 2000.
\newblock ISBN 978-3-540-67993-6.
\newblock \doi{10.1007/BFb0103843}.

\bibitem[Lattimore and Szepesvari(2019)]{lattimore}
T.~Lattimore and C.~Szepesvari.
\newblock \emph{Bandit Algorithms}.
\newblock 2019.

\bibitem[Lovász and Plummer(2009)]{lovasz2009matching}
L.~Lovász and M.~D. Plummer.
\newblock \emph{Matching Theory}, volume 367.
\newblock American Mathematical Society, 2009.

\bibitem[Mahdian and Yan(2011)]{mahdian_yan}
M.~Mahdian and Q.~Yan.
\newblock Online bipartite matching with random arrivals: An approach based on strongly factor-revealing lps.
\newblock In \emph{Proceedings of the Forty-Third Annual ACM Symposium on Theory of Computing}, STOC '11, page 597–606, New York, NY, USA, 2011. Association for Computing Machinery.

\bibitem[Manshadi et~al.(2012)Manshadi, Gharan, and Saberi]{Manshadi_sto}
V.~H. Manshadi, S.~O. Gharan, and A.~Saberi.
\newblock Online stochastic matching: Online actions based on offline statistics.
\newblock \emph{Mathematics of Operations Research}, 37\penalty0 (4):\penalty0 559--573, 2012.
\newblock ISSN 0364765X, 15265471.

\bibitem[Mastin and Jaillet(2013{\natexlab{a}})]{mastin}
A.~Mastin and P.~Jaillet.
\newblock Greedy online bipartite matching on random graphs.
\newblock 07 2013{\natexlab{a}}.

\bibitem[Mastin and Jaillet(2013{\natexlab{b}})]{mastin_jaillet}
A.~Mastin and P.~Jaillet.
\newblock Greedy online bipartite matching on random graphs.
\newblock \emph{ArXiv}, abs/1307.2536, 2013{\natexlab{b}}.

\bibitem[McPherson et~al.(2001)McPherson, Smith-Lovin, and Cook]{mcpherson2001birds}
M.~McPherson, L.~Smith-Lovin, and J.~M. Cook.
\newblock Birds of a feather: Homophily in social networks.
\newblock \emph{Annual Review of Sociology}, 27:\penalty0 415--444, 2001.

\bibitem[Mehta(2013{\natexlab{a}})]{mehta}
A.~Mehta.
\newblock 2013{\natexlab{a}}.
\newblock \doi{10.1561/0400000057}.

\bibitem[Mehta(2013{\natexlab{b}})]{mehta_survey}
A.~Mehta.
\newblock Online matching and ad allocation.
\newblock 8 (4):\penalty0 265--368, 2013{\natexlab{b}}.

\bibitem[Noiry et~al.(2021)Noiry, Perchet, and Sentenac]{noiry2021online}
N.~Noiry, V.~Perchet, and F.~Sentenac.
\newblock Online matching in sparse random graphs: Non-asymptotic performances of greedy algorithm.
\newblock In A.~Beygelzimer, Y.~Dauphin, P.~Liang, and J.~W. Vaughan, editors, \emph{Advances in Neural Information Processing Systems}, 2021.

\bibitem[Schrijver(2003)]{schrijver2003combinatorial}
A.~Schrijver.
\newblock \emph{Combinatorial Optimization: Polyhedra and Efficiency}, volume~24.
\newblock Springer Science \& Business Media, 2003.

\bibitem[Slivkins(2019)]{Slivkins}
A.~Slivkins.
\newblock Introduction to multi-armed bandits.
\newblock \emph{Found. Trends Mach. Learn.}, 12\penalty0 (1–2):\penalty0 1–286, Nov. 2019.
\newblock ISSN 1935-8237.
\newblock \doi{10.1561/2200000068}.

\bibitem[Soprano-Loto et~al.(2023)Soprano-Loto, Jonckheere, and Moyal]{sopranoloto2023online}
N.~Soprano-Loto, M.~Jonckheere, and P.~Moyal.
\newblock Online matching for the multiclass stochastic block model.
\newblock \emph{arXiv preprint arXiv:2303.15374}, 2023.

\bibitem[Warnke(2019)]{Warnke2019OnWD}
L.~Warnke.
\newblock On wormald's differential equation method.
\newblock \emph{ArXiv}, abs/1905.08928, 2019.

\bibitem[Wormald(1995)]{Wormald95}
N.~C. Wormald.
\newblock Differential equations for random processes and random graphs.
\newblock \emph{Annals of Applied Probability}, 5:\penalty0 1217--1235, 1995.

\bibitem[Zdeborová and Mézard(2006)]{Zdeborova_2006}
L.~Zdeborová and M.~Mézard.
\newblock The number of matchings in random graphs.
\newblock \emph{Journal of Statistical Mechanics: Theory and Experiment}, 2006\penalty0 (05):\penalty0 P05003, may 2006.
\newblock \doi{10.1088/1742-5468/2006/05/P05003}.

\end{thebibliography}
\newpage
\appendix
\section{Differential inclusions}
\label{sec:differential_inclusion}
This section aims to introduce the fundamental concepts of differential inclusions. Unlike ordinary differential equations (ODE), where the derivative of the unknown function is determined by a single-valued map, differential inclusions generalize this by allowing the derivative to lie within a set-valued map. This broader framework is well-suited for modeling dynamical systems that involve uncertainty, discontinuities, or control constraints.
\subsection{Set-Valued Maps}
Let $ F : \lR^n \to \lR^n $ denote a set-valued map, i.e a mapping which assigns to each point $ x \in \lR^n $ a subset $ F(x) \subseteq \lR^n $. We consider in general that $ F(x) $ is nonempty for all $ x \in \lR^n $. The notation $\langle x,y\rangle$ denotes the standard inner product on $\lR^{d}$, and the norm is given by $\|x\|=\sqrt{\langle x,x\rangle}$. For a set $A\subset \lR^{d}$, we define its norm as $\|A\|=\sup_{x\in A} \|x\|$.

A set-valued map $ F $ is said to be:
\begin{itemize}
    \item \textbf{Upper semicontinuous (u.s.c.):} A set-valued map $F : \mathbb{R}^n \to \mathbb{R}^n$ is upper semicontinuous at a point $y \in \mathbb{R}^n$ if for every sequence $y^{(n)} \to y$ and any sequence $x_n \in F(y^{(n)})$ such that $x_n \to x$, it holds that $x \in F(y)$.

    \item \textbf{Locally bounded} if for every compact set $ K \subset \mathbb{R}^n $, there exists a constant $M$ such that,
$$
\sup_{x \in K} \sup_{v \in F(x)} \|v\| \leq M.
$$
  
    \item \textbf{Measurable} if its graph is measurable in the product sigma-algebra on $ \lR^n \times \lR^n $.
    \item \textbf{One sided Lipschitz} with constant $L$ if for all $z,\tilde{z} \in \lR^n$ and for $z\in F(y), \tilde{z}\in F(\tilde{y})$, 
    $$\langle y-\tilde{y},z-\tilde{z} \rangle \leq L ||y-\tilde{y}||^2$$
\end{itemize}
\subsection{Definition of differential inclusions}
\begin{definition}
    Let $F: \lR^n \to \lR^n $ be a set valued and $T>0$. A differential inclusion is defined as, 
    \begin{align}
    \label{eq:diffinclusion}
        \dot{x}(t) \in F(x(t)) ~~~~~~\text{for}~t\in [0,T]
    \end{align}
    together with an initial condition, 
    $$x(0)=x_0$$
\end{definition}
Let $I \subseteq \lR$, a function $x: I \to \lR^n$ is a solution to the differential inclusion defined in \Cref{eq:diffinclusion} with initial condition $x(0)=x_0$ if there exists a function $\phi : I \to \lR^n$ such that:
\begin{itemize}
    \item For all $t\in I$: $x(t)=x(0)+\int_{0}^{t} \phi(s)\mathrm{d}s$. 
    \item For almost every $t\in I$ $\phi(t)\in F(x(t))$. 
\end{itemize}
\subsection{Existence and uniqueness of the solution}
\begin{proposition*}(\cite{filippov1988,aubin1984,kunze2000})
\begin{itemize}
    \item If $F$ is upper semicontinuous and if there exists $c$ such that $\|F(y)\|\leq c(1+\|y\|) $ then for any initial condition $x_0$, $\dot{x}\in F(x)$ has at least one solution on $[0,+\infty)$ with $x(0)=x_0$. 
    \item If $F$ is one-sided Lipschitz, then for all $T>0$ there exists at most one solution of $\dot{x}\in F(x)$ on $[0,T]$
\end{itemize}
If $F$ is upper semicontinuous and one-sided Lipschitz then for $T>0$,  $\dot{x}\in F(x)$ has a unique solution on $[0,T]$.
\end{proposition*}

\section{Myopic algorithm}
This section is organized into two parts. In the first part, we analyze the $\myopic$ algorithm. In the second part, we show that, under certain assumptions, the model reduces to the Erdős–Rényi case studied in~\cite{mastin_jaillet}. We begin by considering the following $\myopic$ algorithm, as introduced in the main paper:

\begin{algorithm}[H]
\caption{$\myopic$ policy}
\SetAlgoNlRelativeSize{0}
\DontPrintSemicolon
\KwOut{Updated matching $M(t)$}
Compute the optimal transport plan $Q^*$\;
\For{$t \in [T]$}{
Choose $c_t \in \mathcal{C}$ at random with probability $Q^*(c_t,d_t)/\nu(d_t)$

\If{$\mathcal{F}_{c_t}(t) = \emptyset$}{
    $M(t) = M(t-1)$\;
}
\Else{
    $M(t) = M(t-1) \cup \{(n_t, t)\}$ for $n_t \sim \text{unif}(\mathcal{F}_{c_t}(t))$
    }}
\end{algorithm}
\subsection{Proof of \Cref{theorem:approxsolution} }
\label{sec:proof_th1}
\approxsolution*
The proof of \Cref{theorem:approxsolution} is based on the Wormald theorem \cite{Warnke2019OnWD,Wormald95} and is structured as follows:
\begin{itemize}
\item We first define the evolution of $M_c(t)$, the size of the matching constructed by $\myopic$ in class $c \in [C]$ at time $t \in [T]$. 
\item We then verify that $M_c(t)$ satisfies the conditions required to apply the Wormald theorem \cite{Warnke2019OnWD,Wormald95}. \item We apply the Wormald theorem \cite{Warnke2019OnWD,Wormald95} to analyze the behavior of $M_c(t)$. 
\item Next, we construct an approximate solution to the differential equation that serves as a continuous approximation of $M_c(t)$. 
\item Finally, we derive an explicit bound on the error between the true solution and its approximation. 
\end{itemize}
Let $\mathbf{M}(t) = (M_1(t), \ldots, M_C(t))$ denote the vector of matching sizes in each class $c \in [C]$ constructed by the $\myopic$ algorithm. For each class $c \in [C]$, the matching size evolves according to the following dynamics:
\begin{align} M_c(t+1) = M_c(t) + \indicator{\exists n \in \mathcal{N}_c(t) ~\text{s.t}~ c_t^* = c ~\text{and}~m_n(t+1) = 1} \end{align}

Here, $c_t^*$ represents the class selected by $\myopic$ at time $t$.

The first step is to compute the expected one-step change in $M_c(t)$, the size of the matching constructed by $\myopic$. This is formalized in the following lemma.
\begin{restatable}{lemma}{onestepchangeSi}\label{lemma:one_step_change_Si}
For $t\in [T]$, $c\in [C]$ and for any $B=(b_1,\hdots,b_{C})$, the expectation of the one-step change of $M_c(t)$, when matching is constructed using $\myopic$ is given by,
    \begin{equation}
        \lE[M_c(t+1)-M_c(t)|\mathbf{M}(t), B]= F_c(t,M_1(t),\hdots, M_C(t))
    \end{equation}
    where $F_c(t,M_1(t),\hdots, M_C(t))=\sum_{d=1}^{D}\left(1 - \left(1-\frac{a_{c,d}}{N}\right)^{Nb_c-M_c(t)}\right) Q^*(c,d) $
\end{restatable}

\begin{proof}
    
    Moving to conditional expectation gives, 
    \begin{align}
        \lE&[M_c(t+1)-M_c(t)|\mathbf{M}(t), B]\\
        &= \lE[\indicator{\exists n\in {\cal N}_{c}(t),c_t^*=c, m_n(t+1)=1}|\mathbf{M}(t), B]\\
        &= \lP(\exists n\in {\cal N}_{c}(t),c_t^*=c, m_n(t+1)=1|\mathbf{M}(t), B)\\
        &= \sum_{d=1}^{D}\lP(\exists n\in {\cal N}_{c}(t),c_t^*=c, m_n(t+1)=1|\mathbf{M}(t), B, d(t+1)=d)\nu(d)\\
        &= \sum_{d=1}^{D}\lP(\exists n\in {\cal N}_{c}(t), m_n(t+1)=1|\mathbf{M}(t), B, d(t+1)=d,c_t^{*}=c) \nu(d)\nonumber\\
        &~~~~~~~~~~~~~~ \lP(c_t^*=c|\mathbf{M}(t), B, d(t+1)=d)\\
        &= \sum_{d=1}^{D}\left(1 - (1-p(c,d))^{Nb_c-M_c(t)}\right) \nu(d)Q^*(c,d)/\nu(d)  \\
        &= \sum_{d=1}^{D}\left(1 - \left(1-\frac{a_{c,d}}{N}\right)^{Nb_c-M_c(t)}\right)  Q^*(c,d)
        \end{align}
\end{proof}
The following lemma establishes the Lipschitz continuity of the function $f_c(x)= \sum_{d=1}^{D} (1-e^{-a_{c,d}(Nb_c-x)})Q^{*}(c,d)$. 
\begin{lemma}
\label{lemma:Lipschitzconstantgreedy}
    For $x\leq Nb_c$ and $c\in [C]$, the function $f_c$ is $L_c$-Lipschitz with  $L_c= \sum_{d=1}^{D} a_{c,d} Q^{*}(c,d) $.
\end{lemma}
\begin{proof}
   For $x \leq Nb_c$, the function $f_c$ is a sum of continuous and differentiable functions on $\lR$. Its derivative is given by
\begin{align}
    |f_c'(x)| = \sum_{d=1}^{D} e^{-a_{c,d}(Nb_c - x)} a_{c,d} Q^{*}(c,d) \leq \sum_{d=1}^{D} a_{c,d} Q^{*}(c,d) .
\end{align}
Since $f_c$ is differentiable, the Mean Value Theorem implies that for any $x, y \in \lR$ with $x, y \leq Nb_c$, there exists $\xi \in (x, y)$ such that
\begin{align}
    |f_c(x) - f_c(y)| = |f_c'(\xi)| \, |x - y| \leq L_c |x - y|,
\end{align}
where the Lipschitz constant is defined by $ L_c = \sum_{d=1}^{D} a_{c,d} Q^{*}(c,d) $.
\end{proof}
The following lemma is a technical lemma, 
 \begin{lemma}
\label{lemma:technical_lemma}
    For $n>0$, $a\leq n/2$ and $0\leq w\leq 1$, 
    $$0\leq e^{-aw}-\left(1-\frac{a}{n}\right)^{nw}\leq \frac{a}{ne}$$
\end{lemma}
\begin{proof}
    Using the following inequalities: $1-x \geq e^{-x-x^2}$ for $x\leq \frac{1}{2}$ and $1-x\leq e^{-x}$ for $x\geq 0$, we obtain $e^{-aw}\left(1-\frac{a^2 w}{n}\right)\leq \left(1-\frac{a}{n}\right)^{nw}\leq e^{-aw} $.  The result follows by rearranging terms and using that $aw e^{-aw}\leq 1/e$. 
    \end{proof}
    
The next lemma bounds the distance between $F_c$ defined in \Cref{lemma:one_step_change_Si} and $f_c$, 
\begin{lemma}
\label{lemma:distance_f_cF}
    For $c\in [C]$, 
    \begin{align}
        |F_c(t,M_1(t),\hdots,M_{C}(t))- f_c(M_c(t))|\leq \sum_{d=1}^{D}\frac{a_{c,d}}{N e}Q^{*}(c,d)
    \end{align}
\end{lemma}
\begin{proof}
    For $c\in [C]$, 
    \begin{align}
       &|F_c(t,M_1(t),\hdots,M_{C}(t))- f_c(M_c(t))|\\
       &= \left|\sum_{d=1}^{D} e^{-a_{c,d}(Nb_c-M_c(t))}Q^{*}(c,d)-\left(1-\frac{a_{c,d}}{N}\right)^{Nb_c-M_c(t)}Q^{*}(c,d)\right|\\
       &\leq \sum_{d=1}^{D}\frac{a_{c,d}}{N e}Q^{*}(c,d)~~~~~~~~~~~~(\Cref{lemma:technical_lemma})
    \end{align}
\end{proof}
Now we are ready to prove \Cref{theorem:approxsolution}.
 \begin{proof}
 From \Cref{lemma:one_step_change_Si}, we have 
 \begin{align}
     \lE[M_c(t+1)-M_c(t)|\mathbf{M}(t), B]&= \sum_{d=1}^{D}\left(1 - \left(1-\frac{a_{c,d}}{N}\right)^{Nb_c-M_c(t)}\right)  Q^{*}(c,d)
 \end{align}
Let $Y_c(s) = \frac{M_c(sN)}{N}$ denote the normalized matching size in class $c\in [C]$, and define the vector $\mathbf{Y}(s) = (Y_1(s), \ldots, Y_C(s))$ for $0 \leq s \leq T/N$. Then, we obtain:

 \begin{align}
    \frac{\lE[Y_c(s+1/N)-Y_c(s)|\mathbf{Y}(s), B]}{1/N}&= \sum_{d=1}^{D}\left(1 - \left(1-\frac{a_{c,d}}{N}\right)^{Nb_c-NY_c(s)}\right)  Q^{*}(c,d)
 \end{align}
As $ N \to \infty $, we find:
$$\text{for}~~ s \in \left[\frac{T}{N}\right]~~~\dot{y}_c(s)= \sum_{d=1}^{D}(1-e^{-a_{c,d}(b_c-y_c(s))}Q^{*}(c,d), \text{and}~~~y_c(0)=0$$

 Applying Wormald’s theorem~\cite{Warnke2019OnWD}, and considering the domain $ 0 \leq y_s \leq 1$ with $\beta = 1$ (by the nature of the matching process), we define the Lipschitz constant as $L_c = \sum_{d=1}^{D} a_{c,d} Q^{*}(c,d)$ (see \Cref{lemma:Lipschitzconstantgreedy}). Additionally, we set $\delta = \sum_{d=1}^{D} \frac{a_{c,d}}{N e} Q^{*}(c,d)$ as in \Cref{lemma:distance_f_cF}, which gives $\lambda = N^{-1/3} \sum_{d=1}^{D} a_{c,d} Q^{*}(c,d)$.
Therefore, with probability at least $1-2Ce^{- N^{1/3}L_c^2/ 8\alpha}$
the approximation holds.
    \begin{align}
         |M_c(t)-Ny_c(t/n)|\leq  3e^{L_c\alpha}N^{2/3} L_c
    \end{align}
     Where $y_c$ satisfies for $s\in [T/n]$, 
\begin{align}
\label{eq:edo_greedy1}
    \dot{y}_c(s)&= \sum_{d=1}^{D}\left(1-e^{-a_{c,d}(b_c-y_c(s))}\right) Q^{*}(c,d)\\
    \nonumber y_c(0)&=0
\end{align}

Now we want to approximate the solution of \Cref{eq:edo_greedy1}.  To simplify the analysis, we introduce the change of variable $ z_c(t) = y_c(t)-b_c  $ where $t\in [T/N]$. Under this transformation, \Cref{eq:edo_greedy1} becomes:
\begin{align}
    \begin{cases}
    \dot{z}_c(t) &= \sum_{d=1}^{D} \left(1 - e^{a_{c,d} z_c(t)} \right) Q^*(c,d) \\
    z_c(0) &= -b_c
    \end{cases}
    \label{eq:edo_greedy3}
\end{align}

The right-hand side of \Cref{eq:edo_greedy3} involves terms of the form $ 1 - e^{a_{c,d} z_c} $. Since $ -b_c \leq z_c \leq 0 $, we can use the Taylor expansion:
\[
 1-e^{a_{c,d} z_c}  = -a_{c,d} z_c + \mathcal{O}(z_c^2).
\]

Consider now $ \tilde{z}_c $ be the solution of the following linearized differential equation with initial condition $ \tilde{z}_c(0) = -b_c $ and $t\in [T/N]$:
\begin{align}
\label{eq:edo_greedy_approx}
    \tilde{z}_c'(t) = -\tilde{z}_c(t) \sum_{d=1}^{D} a_{c,d} Q^*(c,d).
\end{align}

The solution to \Cref{eq:edo_greedy_approx} is explicitly given by:
\[
\tilde{z}_c(t) = -b_c \exp\left( -t \sum_{d=1}^{D} a_{c,d} Q^*(c,d) \right).
\]

With $ \tilde{z}_c $  in hand, we are ready to bound the error between the exact solution of \Cref{eq:edo_greedy3} and its approximation $\tilde{z}_c$.
Let $e_c= \tilde{z}_c-z_c$ be the approximation error, its derivative for $t\in [T/N]$ is given by, 
\begin{align}
    e_c'(t)&=  \tilde{z}_c'(t)- z_c'(t)\\
    &= \sum_{d=1}^{D}(-a_{c,d}\tilde{z}_{c}(t)-1+e^{a_{c,d}z_{c}(t)}) Q^{*}(c,d)
    \end{align}
Since $ z_c \in [-b_c, 0] $, we have the inequality
$$
1 - e^{a_{c,d} z_c(t)} \leq -a_{c,d} z_c(t) \quad \text{for all } d = 1, \dots, D.
$$
This implies,
$$
\dot{z}_c(t) = \sum_{d=1}^{D} \left(1 - e^{a_{c,d} z_c(t)}\right) Q^*(c,d)
\leq -z_c(t) \sum_{d=1}^{D} a_{c,d} Q^*(c,d).
$$
This upper bound matches the right-hand side of the linearized system~\Cref{eq:edo_greedy_approx} evaluated at $ z_c(t) $. Since both solutions share the same initial condition, $ z_c(0) = \tilde{z}_c(0) = -b_c $, we can apply the comparison principle~\cite[Chapter~3, Lemma~3.4]{Khalil:1173048}.

It follows that
\[
z_c(t) \leq \tilde{z}_c(t) \quad \text{for all } t \in [T/N],
\]
which implies that the approximation error
\[
e_c(t) := \tilde{z}_c(t) - z_c(t) \geq 0.
\] 
 For all $t\in [T/N]$, using Taylor expansion of order 2, 
    \begin{align}
        e_c'(t)&\leq \sum_{d=1}^{D} \left(-a_{c,d}\tilde{z}_c(t)+ a_{c,d}z_c(t)+\frac{a_{c,d}^{2}z_c^{2}(t)}{2}\right)Q^{*}(c,d)\\
        e_c'(t)&\leq \sum_{d=1}^{D} \left(-a_{c,d}e_{c}(t)+\frac{a_{c,d}^{2}b_c^{2}}{2}\right)Q^{*}(c,d) ~~~~~~~~~~~~~~(\text{using}~z_c\in[-b_c,0] )\\
        e_c'(t)&\leq -e_c(t)L_c+ J_c \label{eq:28}
    \end{align}
   We multiply  \Cref{eq:28} by $e^{L_c t}$, 
    
    \begin{align}
        e_c'(t)e^{L_c t} + L_ce^{L_c t}e_c(t)&\leq J_c e^{L_c t}
    \end{align}
    Integrating both sides, 
    \begin{align}
        e^{t L_c} e_{c}(t)-e_c(0)\leq J_c \frac{e^{L_c t}-1}{L_c}
    \end{align}
    Thus, 
    \begin{align}
        e_c(t)&\leq \frac{J_c}{L_c}(1-e^{-L_c t}) +e(0)e^{-L_c t}\\
        e_{c}(t)&\leq \frac{J_c}{L_c}(1-e^{-L_c t})
        \label{eq:error_bound_ec_theorem1}
    \end{align}
    Thus, for $t\in [0,\frac{T}{N}]$ , 
    $$z_c(t)= \tilde{z_c}(t)- e_c(t)$$
    where $e_c(t)$ satisfies \Cref{eq:error_bound_ec_theorem1}, with $L_{c}=  \sum_{d=1}^{D}a_{c,d}Q^{*}(c,d)$ and $J_c= \sum_{d=1}^{D} \frac{a_{c,d}b_c^2}{2} Q^{*}(c,d)$.

Thus replacing $z_c(t)$ by $y_c(t)-b_c$ we get the final result. 
\end{proof}
\subsection{Recovering the Erdős–Rényi case}
Consider the special case where the connection probability depends only on the class $ c $, that is, $ p(c,d) = \frac{a_c}{n} $. In this setting, the graph structure intoduced in \Cref{subsec:model} simplifies to an Erdős–Rényi random graph \cite{bollobas2001random,janson2011random}. Under this assumption, \Cref{eq:edo_greedy1} reduces to:
\begin{align}
    \dot{z}_c(t) &= \left(1 - e^{-a_c z_c(t)}\right) \sum_{d=1}^{D} Q^{*}(c,d), \\
    \frac{-a_c \dot{z}_c(t) \, e^{-a_c z_c(t)}}{e^{-a_c z_c(t)} - 1} &= -a_c \sum_{d=1}^{D} Q^{*}(c,d).
\end{align}

Integrating both sides with respect to time yields:
\begin{align}
    \ln\left| e^{-a_c z_c(t)} - 1 \right| - \ln\left| e^{-a_c b_c} - 1 \right| = -a_c t \sum_{d=1}^{D} Q^{*}(c,d).
\end{align}

Solving for $ z_c(t) $, we obtain the closed-form expression:
\begin{align}
\label{eq:z_isol_erdos_reyni}
    z_c(t) = -\frac{1}{a_c} \ln\left(1 + \left(e^{-a_c b_c} - 1\right) e^{-a_c t \sum_{d=1}^{D} Q^{*}(c,d)}\right).
\end{align}

This shows that when the model is reduced to the Erdős–Rényi setting, we obtain an exact solution to the corresponding differential equation. Consequently, with high probability, the size of the matching in each class $ c \in [C] $ concentrates around $ z_c $ as given in \Cref{eq:z_isol_erdos_reyni}. These results are in close agreement with those found in \cite{mastin}, which also examined the dynamics of matching in Erdős–Rényi graphs and observed similar asymptotic behavior.
\section{Balance algorithm}
\label{sec:balance_algo}
This section is organized into two parts. First, we analyze the case where matching is performed using the $\balance$ algorithm, as defined in the main paper. Next, we extend the analysis to the $\realbalance$ algorithm.


\subsection{$\balance$}
\label{subsec:balance}
We consider the following $\balance$ algorithm. 

\begin{algorithm}[H]
\caption{$\balance$}
\SetAlgoNlRelativeSize{0}
\DontPrintSemicolon
\KwOut{Updated matching $M(t)$}
\For{$t \in [T]$}{

Choose $c_t= \argmax_{c\in [1,C]} \sum_{d=1}^{D}(1-(1-\frac{a_{c,j}}{N})^{Nb_c-M_c(t)})\nu(d)$ \; 

\uIf{${\cal F}_{c_t}(t) = \emptyset$}{
    $M(t) = M(t-1)$\;
}
\Else{
        $M(t) = M(t-1) \cup \{(n_t, t)\}$ for $n_t \sim \text{unif}(\mathcal{F}_{c_t}(t))$
}
}
\end{algorithm}
\subsubsection{Proof of \Cref{theorem:inclusiondiffsol}}
\inclusiondiffsol*
The proof of \Cref{theorem:inclusiondiffsol} is structured as follows:
\begin{itemize}
    \item We first characterize the drift of the process $M_c$.
    \item Next, we verify that $ M_c $ satisfies the assumptions required by Theorem 1 in \cite{gast:inria-00491859}.
    \item Finally, we define the associated differential inclusion and apply Theorem 4 from \cite{gast:inria-00491859} to derive an explicit rate of convergence.
\end{itemize}
The following lemma computes the drift of the process $M_c$ for $c\in [C]$, defined as the conditional expectation $\lE[M_c(t+1) - M_c(t)| B, \mathbf{M}(t)] $, where $B = (b_1, \ldots, b_C)$ and $\mathbf{M}(t)=(M_1(t),\hdots,M_C(t))$. 
\begin{lemma}
\label{lemma:onestepchangebalance}
    For $c\in [C]$, 
    $$ \lE[M_c(t+1)-M_c(t)|B,\mathbf{M}(t)]= H_{c,b_{c},N}(M_c(t)) \indicator{\underset{{k\in [C]}}{\max} H_{k,b_{k},N}(M_{k}(t))= H_{c,b_{c},N}(M_{c}(t))}$$
    where $ H_{c,b_{c},N}(x) = \sum_{d=1}^{D} \left( 1 - \left( 1 - \frac{a_{c,d}}{N} \right)^{Nb_c- x} \right) \nu(d)$.
\end{lemma}
\begin{proof}
 For $c\in [C]$, as defined previously, $M_{c}(t)$ is the number of matched nodes in the class $c$, for $t\in [T]$, $M_c(t)$ follows the dynamics, 
\begin{align}
    M_c(t+1)=M_c(t)+\indicator{\exists n\in {\cal N}_{c}(t)~\text{s.t}~c_t^*=c ~\text{and}~ m_n(t+1)=1}
\end{align}
Let $ c_t^* $ denote the class selected by the $\balance$ algorithm. We define the expected one-step change of the process $ M_c(t) $, for each $ c\in [C] $ and $ t \in [T] $, as follows:
\begin{align}
    &\lE[M_c(t+1)-M_c(t)|N,\mathbf{M}(t)]\\
    &=\lP(\exists n\in {\cal N}_{c}(t)~\text{s.t}~c_t^*=c ~\text{and}~ m_n(t+1)=1| \mathbf{M}(t),B)\\
    &=\sum_{d=1}^{D}\lP(\exists n\in {\cal N}_{c}(t)~\text{s.t}~c_t^*=c ~\text{and}~ m_n(t+1)=1|\mathbf{M}(t), B, d(t+1)=d)\nu(d)\\
        &= \sum_{d=1}^{D}\lP(\exists n\in {\cal N}_{c}(t) ~\text{s.t}~ m_n(t+1)=1|\mathbf{M}(t),B, d(t+1)=d,c_t^*=c) \nu(d)\nonumber\\
        &~~~~~~~~~~~~~~ \lP(c_t^*=c|\mathbf{M}(t), R, d(t+1)=d)\\
        &= H_{c,b_{c},N}(M_c(t)) \indicator{\underset{{k\in [C]}}{\max} H_{k,b_{k},N}(M_{k}(t))= H_{c,b_{c},N}(M_{c}(t))}
\end{align}
 where $ H_{c,b_{c},N}(x) = \sum_{d=1}^{D} \left( 1 - \left( 1 - \frac{a_{c,d}}{N} \right)^{Nb_c - x} \right) \nu(d)$.
 \end{proof}
 The next lemma defines a martingale difference sequence based on the process $M_c$, and proves that its second moment is bounded.  
 \begin{lemma}
\label{lemma:martingale}
Let $c \in [C]$ and $t \in [T]$. Define the process $Q_c(t+1) = M_c(t+1) - M_c(t) - \mathbb{E}[M_c(t+1) - M_c(t)|\mathbf{M}(t),B]$. 
Then  $(Q_c(t))_{t \in [T]}$ is a martingale difference sequence with respect to the filtration generated by $\mathbf{M}(t)$. Moreover, there exists a constant $ b > 0 $ such that:
$$\lE[Q_c(t+1)|\mathbf{M}(t),B]=0$$ and $$\lE[|Q_c(t+1)|^2|\mathbf{M}(t),B]\leq b$$
\end{lemma}

 \begin{proof} By direct computation, we obtain $$ \mathbb{E}[Q_c(t+1)|\mathbf{M}(t),B] = 0. $$ Furthermore, from the definition of the matching process, for all $c\in [C]$ we have $$ |M_c(t+1) - M_c(t)| \leq 1, \quad \forall t \in [T]. $$ This implies that the second moment of $ Q_c(t+1) $ is bounded. \end{proof}

 The next result proves the first assumption of Theorem $1$ in \cite{gast:inria-00491859}. 
 \begin{lemma}
 \label{lemma:assum_gast1}
     For $c\in [C]$ and $t\in [T]$, let $\mathcal{H}_{c,b_{c},N}(y)={H}_{c,b_{c},N}(y) \indicator{\underset{{k\in [C]}}{\max} {H}_{k,b_{k},N}(y)= {H}_{c,b_{c},N}(y)}$, it satisfies,
     $$\forall y\leq Nb_c, |\mathcal{H}_{c,b_{c},N}(y)|\leq c(1+|y|) $$
     with $\alpha_c=\ln\left(1-\frac{\min_{j\in [D]}(a_{c,j})}{N}\right) , c= \max(|1-e^{\alpha_c {Nb_c}}|,|\alpha_c e^{\alpha_c {Nb_c}}|)$.
 \end{lemma}
 \begin{proof}
     \begin{align}
         |\mathcal{H}_{c,b_{c},N}(y)| &\leq \left| \sum_{d=1}^{D} \left(1-\left(1-\frac{a_{c,d}}{N}\right)^{Nb_c-y}\right)\nu(d)\right|\\
         &= \left| 1-\sum_{d=1}^{D} e^{\ln\left(1-\frac{a_{c,d}}{N}\right){(Nb_c-y)}}\nu(d)\right|\\
         &\leq \left| 1-\sum_{d=1}^{D} e^{\ln\left(1-\frac{\min_{j}(a_{c,j})}{N}\right){(Nb_c-y)}}\nu(d)\right|\\
         &\leq  \left| 1- e^{\ln\left(1-\frac{\min_{j}(a_{c,j})}{N}\right){(Nb_c-y)}}\right|\\
         &\leq  \left| 1- e^{\alpha_c {Nb_c}}(1+\alpha_c y)\right|\\
         &\leq |1-e^{\alpha_c {Nb_c}}|+ |\alpha_c e^{\alpha_c {Nb_c}} y|\\
         &\leq c (1+|y|)
     \end{align}
     with $\alpha_c=\ln\left(1-\frac{\min_{j}(a_{c,j})}{N}\right) , c= \max(|1-e^{\alpha_c {Nb_c}}|,|\alpha_c e^{\alpha_c {Nb_c}}|)$.
 \end{proof}
 
The following technical lemma provides a bound on the distance between ${H}_{c,b_{c},N}$ and its limit as $N$ becomes large.

\begin{lemma}
For $c\in [C]$,
    $$\left|\sum_{d=1}^{D} \left(1-e^{-a_{c,d}(b_c-M_c(t)/N)}\right)\nu(d)-\left(1-\left(1-\frac{a_{c,d}}{N}\right)^{Nb_c-M_c(t)}\right) \nu(d) \right|\leq \sum_{d=1}^{D}\frac{a_{c,d}}{Ne}\nu(d)$$
\end{lemma}
\begin{proof}
    Let $A=\left|\sum_{d=1}^{D} \left(1-e^{-a_{c,d}(b_c-M_c(t)/N)}\right)\nu(d)-\left(1-\left(1-\frac{a_{c,d}}{N}\right)^{Nb_c-M_c(t)}\right) \nu(d) \right|$, 
    \begin{align}
        A&\leq \sum_{d=1}^{D}\left| \left(1-e^{-a_{c,d}(b_c-M_c(t)/N)}\right)\nu(d)-\left(1-\left(1-\frac{a_{c,d}}{N}\right)^{Nb_c-M_c(t)}\right) \nu(d) \right|\\
        &\leq \sum_{d=1}^{D}\left| e^{-a_{c,d}(b_c-M_c(t)/N)}\nu(d)-\left(1-\frac{a_{c,d}}{N}\right)^{Nb_c-M_c(t)} \nu(d) \right|\\
        &\leq \sum_{d=1}^{D}\frac{a_{c,d}}{Ne} \nu(d)~~~~~~~~~~~~~~~~~~~~(\text{Lemma } \ref{lemma:technical_lemma})
    \end{align}
\end{proof}
The following lemma shows that the function $ f_{c,b_{c}} $, defined as the limit of $ H_{c,b_{c},N} $ as $ N \to \infty $ and given by $f_{c,b_{c}}(x) = \sum_{d=1}^{D} \left(1 - e^{-a_{c,d}(b_c - x)}\right) \nu(d)$, is $ L_c $-Lipschitz continuous.

 \begin{lemma}(Lipschitz condition)
 \label{lemma:fiLipschitz}
 For $c\in [C]$, the function $f_{c,b_{c}}(x)= \sum_{d=1}^{D} (1-e^{-a_{c,d}(b_c-x)})\nu(d) $ is Lipschitz continuous with constant $L_c= \sum_{d=1}^{D} a_{c,d} \nu(d)$. 
 \end{lemma}
 \begin{proof}
     Let $x,y$ such that $x\leq b_c$ and $y\leq b_c$ for all $c\in [C]$, 
     \begin{align}
         |f_{c,b_{c}}(x)-f_{c,b_{c}}(y)|= \left|\sum_{d=1}^{D} (e^{-a_{c,d}(b_c-y)}-e^{-a_{c,d}(b_c-x)})\nu(d) \right|
     \end{align}
     By mean value Theorem, there exists \( \xi \in (x, y) \) such that we have, 
     \begin{align}
         |f_{c,b_{c}}(x)-f_{c,b_{c}}(y)|&\leq  |x-y|\sum_{d=1}^{D} e^{-a_{c,d}(b_c-\xi)}a_{c,d} \nu(d) \\
         &\leq |x-y|\sum_{d=1}^{D} a_{c,d} \nu(d)
     \end{align}
 \end{proof}
The two following lemma provides a technical bound essential for deriving the explicit rate of convergence in \Cref{theorem:inclusiondiffsol}.
For $t\in [T]$ let, 
\begin{align}
    V_c(t)=\frac{1}{N}\sum_{k=0}^{t}Q_c(k+1)
\end{align}
and 
\begin{align}
\label{eq:otherdefMi}
    \frac{M_c(t+1)}{N}&=\frac{M_c(0)}{N}+ \frac{1}{N}\sum_{l=0}^{t}\mathcal{H}_{c,b_{c},N}(M_c(t))+\frac{1}{N}\sum_{l=0}^{t}Q_c(l+1)\\
    &= \frac{1}{N}\sum_{l=0}^{t}\mathcal{H}_{c,b_{c},N}(M_c(t))+\frac{1}{N}\sum_{l=0}^{t}Q_c(l+1)
\end{align}
\begin{lemma}
    \label{lemma:maximal_bound_martingale}
     For all $T,N>0$, and for all $\epsilon>0$, 
    $$\lP\left(\sup_{0\leq k\leq T}|V_c(k)|\geq \epsilon\right)\leq \frac{Tb}{N^2\epsilon^2}$$
\end{lemma}
\begin{proof}
    Since $\lE\left[Q_c(t+1)|\mathbf{M}(t),B\right]=0$ and $\lE\left[|Q_c(t+1)|^2|\mathbf{M}(t),B\right]\leq b$, we have $\lE[|V_c(t)|^2]\leq \frac{t b}{N^2}\leq \frac{Tb}{N^2}$ for all $t\leq T$. Applying Kolmogorov’s inequality (maximal inequality) for martingales to the
martingale $V$ leads to the bound of the lemma. 
\end{proof}
 \begin{lemma}
\label{lemma:18gast}
For $c\in [C]$, let $M_c$ be defined as in \Cref{eq:otherdefMi} with $|\mathcal{H}_{c,b_{c},N}(y)|\leq c(1+|y|)$. Let $y$ denote the solution of the differential equation associated with $F$ that is defined in \Cref{eq:diffinclusion}. 
    
    If we denote by $\epsilon:=\sup_{l\leq t} |V_c(l)|$ , then
    $$\max\left\{\sup_{0\leq t\leq T}|M_c(t)|, \sum_{0\leq \tau\leq T/n}|m(\tau)|\right\} \leq K_{\alpha}$$
    with $K_{\alpha}= \left(c\alpha+\epsilon\right)e^{c\alpha} /c$ with $c$ as defined in \Cref{lemma:assum_gast1}. 
\end{lemma}
\begin{proof}
    By definition of $M_c$ in \Cref{eq:otherdefMi} and $\Cref{lemma:maximal_bound_martingale}$, we have,
    \begin{align}
        |M_{c}(t+1)/N|&\leq \frac{1}{N}\sum_{l=0}^{t} c(1+|M_{c}(l)|) +\epsilon\\
        &= \frac{t c}{N}+\epsilon+ \frac{c}{N} \sum_{l=1}^{t} |M_c(l)|\\
        &\leq (\frac{Tc}{N}+\epsilon)e^{cT/N}/c
    \end{align}
    The final inequality follows from the discrete Gronwall’s lemma. Substituting $T = \alpha N$ then yields the desired result.
\end{proof}
The next lemma proves that $F$ defined in \Cref{theorem:inclusiondiffsol} is upper semicontinuous,
\begin{lemma}(Upper semicontinuous)
\label{lemma:uppersemicontinuous}
   Let $f_{c,b_{c}} : \mathbb{R} \to \mathbb{R}$ be continuous for each $c \in [C]$, and define the set-valued map $F(m) = \operatorname{conv}\left(f_{c,b_{c}}(m_c) e_c \mid c \in \arg\max_{j \in [C]} f_{j,b_{j}}(m_j)\right)$
    where $e_c$ is the $c$-th standard basis vector in $\mathbb{R}^C$. Then $F$ is upper semicontinuous as a set-valued map from $\mathbb{R}^C$ to subsets of $\mathbb{R}^C$.
\end{lemma}
\begin{proof}
Let $(m^{(n)})_{n \in \mathbb{N}} \subset \mathbb{R}^C$ be a sequence such that $m^{(n)} \to m$ as $n \to \infty$, meaning:
\[
\forall c \in [C], \quad m^{(n)}_c \to m_c.
\]
Let $x_n \in F(m^{(n)})$ be such that $x_n \to x \in \mathbb{R}^C$. We aim to show that $x \in F(m)$.

Each $x_n \in F(m^{(n)})$ belongs to the convex hull:
\[
F(m^{(n)}) = \operatorname{conv}\left(f_{c,b_{c}}(m^{(n)}_c) e_c \mid c \in \arg\max_{j \in [C]} f_{j,b_{j}}(m^{(n)}_j)\right).
\]
Because the index set $[C]$ is finite, there are only finitely many possible argmax sets. Thus, we may extract a subsequence (still denoted by $n$ for simplicity) such that for some fixed $A \subseteq [C]$,
\[
\arg\max_{j \in [C]} f_{j,b_{j}}(m^{(n)}_j) = A \quad \text{for all large } n.
\]

Since each $f_{j,b_{j}}$ is continuous and $m^{(n)}_j \to m_j$, we have $f_{j,b_{j}}(m^{(n)}_j) \to f_{j,b_{j}}(m_j)$ for all $j$, and hence:
\[
\max_{j \in [C]} f_{j,b_{j}}(m^{(n)}_j) \to \max_{j \in [C]} f_j(m_j).
\]
Therefore, for every $i \in A$,
\[
f_{c,b_{c}}(m_c) = \max_{j \in [C]} f_{j,b_{j}}(m_j),
\]
i.e., $A \subseteq \arg\max_{j \in [C]} f_{j,b_{j}}(m_j)$.

Now, each $x_n$ is a convex combination of the vectors $f_{c,b_{c}}(m^{(n)}_c) e_c$ with $c \in A$, and by continuity:
\[
f_{c,b_{c}}(m^{(n)}_c) \to f_{c,b_{c}}(m_c), \quad \text{so} \quad f_{c,b_{c}}(m^{(n)}_c) e_c \to f_{c,b_{c}}(m_c) e_c.
\]
Thus, $x_n \to x$ implies that $x$ lies in the convex hull of the limit points:
\[
x \in \operatorname{conv}\left(f_{c,b_{c}}(m_c) e_c \mid c \in A\right) \subseteq F(m),
\]
since $A \subseteq \arg\max_{j \in [C]} f_{j,b_{j}}(m_j)$.

Hence, every limit point of a converging sequence $(x_n)$ with $x_n \in F(m^{(n)})$ lies in $F(m)$, which proves that $F$ is upper semicontinuous.
\end{proof}
With all the preparatory results established — in particular, \Cref{lemma:assum_gast1,lemma:martingale}, which shows that $M_c(t)$ satisfies the assumptions of Theorem 1 in \cite{gast:inria-00491859} — we are now ready to prove \Cref{theorem:inclusiondiffsol}.

\begin{proof}
For all $c\in [C]$ and $\tau\in [T/n]$, we consider the process $\tilde{M}_c(\tau)$ defined by,
\begin{align}
\label{eq:processmatchdiffinclu}
    \tilde{M}_c(\tau+1/n)= \tilde{M}_c(\tau) +\frac{1}{N}\mathcal{H}_{c,b_{c},N}(\tilde{M}_c(\tau))+\tilde{Q}_c(\tau+1/n)
\end{align}
where $\tilde{Q}_c(\tau)=\frac{Q_c(\tau n)}{n}$ with $Q_c$ as defined in \Cref{lemma:martingale},$\mathcal{H}_{c,b_{c},N}$ as defined in \Cref{lemma:assum_gast1}, and let $\tilde{\mathbf{M}}(\tau)=(\tilde{M}_1(\tau),\hdots,\tilde{M}_{C}(\tau))$.

When $N \to \infty$, and $t\in [T]$, the function ${H}_{c,b_{c},N}$ converges to:
$$
    {H}_{c,b_{c},N}(M_c(t)) \to \sum_{d=1}^{D} \left( 1-e^{-a_{c,d}(b_c- \frac{M_c(t)}{N})} \right) \nu(d) = f_{c,b_{c}}(M_c(t)/N).
$$

Let $g(\tilde{\mathbf{M}})=\left(f_{c,b_{c}}(\tilde{M}_c) \mathbf{1}\{\underset{{k\in [C]}}{\max} f_{k,b_{k}}(\tilde{M}_k)= f_{c,b_{c}}(\tilde{M}_c)\}\right)_{c\in [C]}$ be the drift vector. According to Lemma \ref{lemma:assum_gast1}, each element of the drift vector satisfies the first assumption of Theorem 1 in \cite{gast:inria-00491859}. Moreover, according to Lemma \ref{lemma:martingale}, $\tilde{Q}_{c}(\tau+1/n)$ satisfies the second assumption of Theorem 1 in \cite{gast:inria-00491859}. Since $\tilde{M}_c(0)=0$, based on Theorem 1 in  \cite{gast:inria-00491859}, for all $\tau \in [T/n]$, $\tilde{\mathbf{M}}(\tau)$ converges to $m(\tau)$, where $m$ is the solution of the following differential inclusion:
\begin{align}
\label{eq:diffinclusion}
     \dot{m}(\tau) \in F(m(\tau))
\end{align}
where $F(m)= \operatorname{conv}\left(f_{c,b_{c}}(m_c) e_c  \mid c\in \arg\max_{j\in [C]} f_{j,b_{j}}(m_j)\right)$, and $\operatorname{conv}$ denotes the convex hull.

According to \Cref{lemma:uppersemicontinuous} and \Cref{lemma:assum_gast1}, the differential inclusion \Cref{eq:diffinclusion} admits at least one solution $m$. To establish the uniqueness of this solution, it is sufficient to show that the set-valued map $F$ is one-sided Lipschitz. By \Cref{lemma:fiLipschitz}, each function $f_{c,b_{c}}$ is $L_c$-Lipschitz continuous for all $c \in [C]$, and let $L = \max_{c \in [C]} L_c$. Let $s, s' \in \mathbb{R}^C$ and suppose $z \in F(s)$, $z' \in F(s')$. Then:
\begin{align}
    \langle z - z', s - s' \rangle = \sum_{i=1}^{C} (s_i - s_i') \left(f_{i,b_{i}}(s_i) - f_{i,b_{i}}(s_i')\right) \leq \sum_{i=1}^{C} L_i (s_i - s_i')^2 \leq L \|s - s'\|^2.
\end{align}
Thus $F$ is one-sided Lipschitz with constant $L$, which guarantees the uniqueness of the solution $m$ to the differential inclusion \Cref{eq:diffinclusion}. To get the explicit rate of the convergence of $\tilde{M}_c$ to $m_c$, we use Theorem $4$ of \cite{gast:inria-00491859}. According to \Cref{lemma:martingale}, for $t\in [T]$, $\lE[|Q_c(t+1)|^2|\mathbf{M}(t),B]\leq b$ and $F$ is one sided Lipschitz with constant $L= \max_{c\in [C]} L_c$ where $L_c$ is defined in \Cref{lemma:fiLipschitz}. Thus according to theorem 4 in \cite{gast:inria-00491859}, taking $\delta_c= \sum_{d=1}^{D}\frac{a_{c,d}}{Ne}\nu(d)$, $K_{\alpha}$ as defined in \Cref{lemma:18gast} and $U_c= \sup_{0\leq t\leq T}f_{c,b_{c}}(M_c(t))$, taking $U_c= \sum_{d=1}^{D}(1-e^{-a_{c,d}b_c})\nu(d)$ we define $A_{\alpha,c}= U_c(U_c^2+\frac{14U_c}{3}+2K_{\alpha}), B_{\alpha,c}= 2U_c^2+4L\delta+12K_{\alpha}, C_{\alpha,c}= 2U_c^2+4L\epsilon +8K_{\alpha}$. We have, 
    \begin{align}
        \lP\left(\sup_{0\leq t\leq T} \left|\frac{M_c(t)}{N} -m_c(t/N)\right|\geq \min(\alpha, e^{L\alpha}/\sqrt{2L})\sqrt{A_{\alpha,c}/N+ \delta_c B_{\alpha,c}+\epsilon C_{\alpha,c}}\right)\leq \frac{b T}{N^2\epsilon^2}
    \end{align}
\end{proof}
\subsection{Proof of \Cref{theorem:explicitsolution}}
\label{subsec:proof_solinclusiondiff}
\explicitsolution*
We need a few lemmas before going into the proof of the \Cref{theorem:explicitsolution}.

\begin{restatable}{lemma}{explicitsolutionpiecewise}\label{lem:explicitsolutionpiecewise}
$\forall t\in\mathbb{R}_+, \forall c\in\cC$,
\begin{align}
\mu_c(t) = 
\begin{cases}
    b_c - \bm{\beta}^{(k_t)}_c + f_{c, \bm{\beta}^{(k_t)}}^{-1}(F_{k_t, \bm{\beta}^{(k_t)}}(\mu_{k_t,\bm{\beta}^{(k_t)}} (t - t_{k_t}))) & \text{ if }c \leq k_t\\
    0 & \text{ otherwise}
\end{cases}
\end{align}
\end{restatable}
\begin{proof}
    For $t\geq 0$ and $c > k_t$, $(b_c - \bm{\beta}^{(k_t)}_c) = 0$ by def of $\bm{\beta}^{(k_t)}_c$. Then

    By definition of $k_t$, $t \in [t_{k_t}, t_{k_t + 1})$.
    \begin{align}
        t < t_{k_t + 1} 
        & \Leftrightarrow t < t_{k_t} + \mu_{k_t,\bm{\beta}^{(k_t)}}^{-1}\left(F_{k-1,\bm{\beta}^{(k_t)}}^{-1}\left(f_{k_t+1,\bm{\beta}^{(1)}}(0)\right)\right)\\
        &\Leftrightarrow \mu_{k_t,\bm{\beta}^{(k_t)}} (t - t_{k_t}) < F_{k_t, \bm{\beta}^{(k_t)}}^{-1}\left(f_{k_t+1, \bm{\beta}^{(k_t)}}(0)\right) & \text{ ($\mu_{k_t,\bm{\beta}^{(k_t)}}$ is increasing)}\\
        &\Leftrightarrow F_{k_t, \bm{\beta}^{(k_t)}}(\mu_{k_t,\bm{\beta}^{(k_t)}} (t - t_{k_t})) > f_{k_t+1, \bm{\beta}^{(k_t)}}(0) & \text{ ($F_{k_t, \bm{\beta}^{(k_t)}}$ is decreasing)}\\
        &\Leftrightarrow \forall c > k_t, ~ F_{k_t, \bm{\beta}^{(k_t)}}(\mu_{k_t,\bm{\beta}^{(k_t)}} (t - t_{k_t})) > f_{c, \bm{\beta}^{(k_t)}}(0) & \text{ ($F_{k_t, \bm{\beta}^{(k_t)}}$ is decreasing)}\\
        &\Leftrightarrow \forall c > k_t, ~ f_{c, \bm{\beta}^{(k_t)}}^{-1}(F_{k_t, \bm{\beta}^{(k_t)}}(\mu_{k_t,\bm{\beta}^{(k_t)}} (t - t_{k_t}))) < 0& \text{ ($f_{c, \bm{\beta}^{(k_t)}}^{-1}$ is decreasing)}
    \end{align}
\end{proof}

\begin{restatable}{lemma}{explicitsolutionsumswell}\label{lem:explicitsolutionsumswell}
$\forall t\in\mathbb{R}_+, ~\sum_{c\in\cC} \mu_c(t) = \|b - \bm{\beta}^{(k_t)}\|_1 + \mu_{k_t,\bm{\beta}^{(k_t)}}\left(t - t_{k_t}\right)$
\end{restatable}

\begin{proof}
By definition, for any $t\geq0$,  $\bm{\beta}^{(k_t)}_c \leq b_c$, so $\sum_{c\in\cC} {(b_c - \bm{\beta}^{(k_t)}_c)} = \|b - \bm{\beta}^{(k_t)}\|_1$. Let $k\in[C]$ be a phase.
\begin{align}
    \forall t &\in [t_k, t_{k+1}),~~ \sum_{c\in\cC} \left(f_{c, \bm{\beta}^{(k)}}^{-1}(F_{k, \bm{\beta}^{(k)}}(\mu_{k,\bm{\beta}^{(k)}} (t - t_{k})))\right)_+ = \mu_{k,\bm{\beta}^{(k)}}\left(t - t_{k}\right) \\
    & \Leftrightarrow \forall t\in[0,t_{k+1}-t_k], ~\sum_{c \leq k} f_{c, \bm{\beta}^{(k)}}^{-1}(F_{k, \bm{\beta}^{(k)}}(\mu_{k,\bm{\beta}^{(k)}} (t))) = \mu_{k,\bm{\beta}^{(k)}}\left(t\right) \\
    & \Leftrightarrow \forall s\in\left[\mu_{k,\bm{\beta}^{(k)}}^{-1}(0),\mu_{k,\bm{\beta}^{(k)}}^{-1}(t_{k+1}-t_k)\right], ~\sum_{c\leq k} f_{c, \bm{\beta}^{(k)}}^{-1}(F_{k, \bm{\beta}^{(k)}}(s)) = s \\
    & \Leftrightarrow \forall u\in\left[F_{k, \bm{\beta}^{(k)}}^{-1}\left(\mu_{k,\bm{\beta}^{(k)}}^{-1}(t_{k+1}-t_k)\right),F_{k, \bm{\beta}^{(k)}}^{-1}\left(\mu_{k,\bm{\beta}^{(k)}}^{-1}(0)\right)\right], ~\sum_{c \leq k} f_{c, \bm{\beta}^{(k)}}^{-1}(u) = F_{k, \bm{\beta}^{(k)}}^{-1}(u)
\end{align}
The last line is TRUE by definition of $F_{k, \bm{\beta}^{(k)}}$.
\end{proof}

\begin{restatable}{lemma}{explicitsolutionallequal}\label{lem:explicitsolutionallequal}
$\forall t\in\mathbb{R}_+, \forall c\in\cC$
\begin{align}
f_c(\mu_c(t)) = 
\begin{cases}
F_{k_t, \bm{\beta}^{(k_t)}}(\mu_{k_t,\bm{\beta}^{(k_t)}} (t - t_{k_t})) & \text{ if $c\leq k_t$}\\
f_{c, Nb_c}(0) & \text{otherwise}
\end{cases}
\end{align}
and $f_c(\mu_c(t)) \leq f_{k_t}(\mu_{k_t}(t))$.
\end{restatable}

\begin{proof}
For any $t\in\mathbb{R}_+$ and $c\in\cC$,
    \begin{align}
    f_c(\mu_c(t)) 
    & = f_{c, b_c}(\mu_c(t))\\
    & = f_{c, \bm{\beta}_c^{(k_t)}}\left((\bm{\beta}_c^{(k_t)} - b_c + \mu_c(t)\right)\\
    & = f_{c, \bm{\beta}_c^{(k_t)}}\left(\left(f_{c, \bm{\beta}^{(k_t)}}^{-1}(F_{k_t, \bm{\beta}^{(k_t)}}(\mu_{k_t,\bm{\beta}^{(k_t)}} (t - t_{k_t})))\right)_+\right)
\end{align}
\Cref{lem:explicitsolutionpiecewise} allows to conclude.
\end{proof}
With all the preparatory lemmas established, we now proceed to the proof.
\begin{proof}
The result is proven by induction. For $k\in[C]$, the induction hypothesis is

$\bm{(A_k)}~~\forall t \leq t_k,~ \mu$ is such that $\dot{\mu}(t) \in F(\mu(t))$.

For $t\in [t_k, t_{k+1})$, we have $k=k_t$. 
Thus, restricted to the interval $[t_k, t_{k+1})$,
\begin{align}
    \dot{\mu} \in F(\mu) & \Leftrightarrow \dot{\mu} \in {\rm conv}\left(f_c(\mu_c)e_c | c \in \argmax_{k\in[C]}f_k(\mu_k)\right)\label{eq:mainproof1}\\
    & \Leftrightarrow \dot{\mu} \in F_{k_t, \bm{\beta}^{(k_t)}}(\mu_{k_t,\bm{\beta}^{(k_t)}} (t - t_{k_t})) {\rm conv}\left(e_c | c \in \argmax_{k\in[C]}f_k(\mu_k)\right)\label{eq:mainproof2}\\
    & \Leftrightarrow \left(\sum_{c\leq k} \dot{\mu}_c \leq F_{k_t, \bm{\beta}^{(k_t)}}(\mu_{k_t,\bm{\beta}^{(k_t)}} (t - t_{k_t}))\right) \text{ and }\left(\forall c>k, \dot{\mu}_c = 0\right)\label{eq:mainproof3}\\
    & \Leftrightarrow \left(\dot{\mu}_{k_t,\bm{\beta}^{(k_t)}}\left(t - t_{k_t}\right) \leq F_{k_t, \bm{\beta}^{(k_t)}}(\mu_{k_t,\bm{\beta}^{(k_t)}} (t - t_{k_t}))\right)\text{ and }\left(\forall c>k, \dot{\mu}_c = 0\right)\label{eq:mainproof4}
\end{align}
Going from \eqref{eq:mainproof1} to \eqref{eq:mainproof2} comes from \Cref{lem:explicitsolutionallequal}. Going from \eqref{eq:mainproof2} to \eqref{eq:mainproof3} comes from the fact that the convex hull of a set of basis vectors is the $L_1$ ball restricted to the corresponding subspace. Going from \eqref{eq:mainproof3} to \eqref{eq:mainproof4} comes from applying \Cref{lem:explicitsolutionsumswell}. Equation \eqref{eq:mainproof4} is TRUE by definition of $\mu_{k_t,\bm{\beta}^{(k_t)}} (t - t_{k_t})$ for the first term and by applying \Cref{lem:explicitsolutionpiecewise} for the second term.
\end{proof}







\subsection{$\realbalance$}
\label{subsec:real_balance}
While the $\balance$ policy improves upon the purely compatibility-driven $\myopic$ approach by considering the probability of successful matches, it still relies on an expected availability model—it selects the class that likely has unmatched nodes, without verifying their presence. This can lead to wasted opportunities when the chosen class ends up being empty. To address this limitation, we introduce the $\realbalance$ policy, which takes an additional, more grounded step. Instead of only maximizing the expected match probability, $\realbalance$ ensures that the selected class actually contains at least one available node at the time of decision. This additional check prevents the algorithm from targeting unavailable options, making it more efficient. The formal procedure is detailed in Algorithm~\ref{alg:real_balance_algo}.
\begin{algorithm}[H]
\caption{$\realbalance$}
\label{alg:real_balance_algo}
\SetAlgoNlRelativeSize{0}
\DontPrintSemicolon
\KwOut{Updated matching $M(t)$}
\For{$t \in [T]$}{

Choose $c_t= \argmax_{c\in C(t)} \sum_{d=1}^{D}(1-(1-\frac{a_{c,j}}{N})^{Nb_c-M_c(t)})\nu(d)$ \; 

\uIf{${\cal F}_{c_t}(t) = \emptyset$}{
    $M(t) = M(t-1)$\;
}
\Else{
        $M(t) = M(t-1) \cup \{(n_t, t)\}$ for $n_t \sim \text{unif}(\mathcal{F}_{c_t}(t))$
}
}
\end{algorithm}
The $\realbalance$ policy introduces an additional layer of selectivity compared to $\balance$ by explicitly verifying the presence of available nodes in the selected class before committing to a match. While this refinement improves matching efficiency, it also increases the non-smoothness of the system’s dynamics. In particular, the policy induces more discontinuities—not only due to abrupt switches in the maximization rule, but also because the feasibility of a class now depends on the cardinality of its unmatched nodes at each step. As a result, the evolution of the matching process under $\realbalance$ cannot be captured by the same differential inclusion used for $\balance$. Instead, it follows a more constrained inclusion, where the feasible directions of evolution depend explicitly on whether a class still has unmatched capacity. The following theorem formalizes this behavior, showing that, with high probability, the normalized matching sizes converge to the solution of a new differential inclusion that incorporates both compatibility and real-time availability constraints.
\begin{restatable}
    {theorem}{theomatchingbalancegenerale}
    For all $c\in [C]$, the matching size built by $\realbalance$ satisfies for all $t\in [T]$ with probability $1-\frac{b\alpha}{N\epsilon}$, 
 \begin{align}
        \left|\frac{M_c(t)}{N} -m_c(t/N)\right|\leq \min(\alpha, e^{L\alpha}/\sqrt{2L})\sqrt{A_{\alpha,c}/N+ \delta_c B_{\alpha,c}+\epsilon C_{\alpha,c}}
    \end{align}
  where $m_c$ is the $c-$th coordinate of $m$ the unique solution of the differential inclusion $\dot{m}\in G(m)$ and $G(m)=conv\left(f_{c,b_{c}}(m_c)e_c|  c=\argmax_{k\in [1,C]}f_{k,b_{k}}(m_k), b_c>m_c \right)$ with $e_c$ a basis vector of $\lR^{C}$,  $L= \max_{c\in [1,C]} \sum_{d=1}^{D} a_{c,d}\nu(d)$, $\delta_c= \sum_{d=1}^{D}\frac{a_{c,d}}{Ne}\nu(d)$, $K_{\alpha}=(c\alpha+\epsilon)e^{c\alpha}/c$ with $\epsilon$ as defined in \Cref{lemma:18gast} and $c$ in \Cref{lemma:assum_gast1}, $U_c= \sum_{d=1}^{D}(1-e^{-a_{c,d}b_c})\nu(d)$, $A_{\alpha,c}= U_c(U_c^2+\frac{14U_c}{3}+2K_{\alpha}), B_{\alpha,c}= 2U_c^2+4L\delta_c+12K_{\alpha}, C_{\alpha,c}= 2U_c^2+4L\epsilon +8K_{\alpha}$.
\end{restatable}

As in \Cref{subsec:balance}, our objective here is to approximate the matching size $ M_c $ produced by $ \realbalance $ for each class $ c \in [C] $. Following the same approach as before, the first step involves computing the drift of the process $ M_c$.

\begin{lemma}
    For $c\in [C]$, 
    $$ \lE[M_c(t+1)-M_c(t)|B,\mathbf{M}(t)]= H_{c,b_{c},N}(M_c(t)) \indicator{\underset{{k\in [C]}}{\max} H_{k,b_{k},N}(M_{k}(t))= H_{c,b_{c},N}(M_{c}(t))}\indicator{Nb_c>M_c(t)}$$
\end{lemma}
\begin{proof}
\begin{align}
    &\lE[M_c(t+1)-M_c(t)|B,\mathbf{M}(t)]= \lP(\exists n\in {\cal N}_{c}(t)\backslash \mathcal{M}_c(t)(t),c^*=c, m_n(t+1)=1|B,\mathbf{M}(t))\\
    &=\sum_{d=1}^{D}\lP(\exists n\in {\cal N}_{c}(t)\backslash \mathcal{M}_c(t),c^*=c, m_n(t+1)=1|\mathbf{M}(t), B, d(t+1)=d)\nu(d)\\
        &= \sum_{d=1}^{D}\lP(\exists n\in {\cal N}_{c}(t)\backslash \mathcal{M}_c(t), m_n(t+1)=1|\mathbf{M}(t), B, d(t+1)=d,c^*=c) \nu(d)\nonumber\\
        &~~~~~~~~~~~~~~ \lP(c^*=c|\mathbf{M}(t), B, d(t+1)=d)\\
        &= H_{c,b_{c},N}(M_c(t)) \indicator{\underset{{k\in [C]}}{\max} H_{k,b_{k},N}(M_{k}(t))= H_{c,b_{c},N}(M_{c}(t))}\indicator{Nb_c>M_c(t)}
\end{align}
 where $ H_{c,b_{c},N}(x) = \sum_{d=1}^{D} \left( 1 - \left( 1 - \frac{a_{c,d}}{N} \right)^{Nb_c - x} \right) \nu(d)$.
 \end{proof}
\begin{proof}
For all $c\in [C]$ and $\tau\in [T/N]$, we consider the process $\tilde{M}_c(\tau)$ defined by,
$$
    \tilde{M}_c(\tau+1/N)=\tilde{M}_c(\tau) +\frac{1}{N}H_{i,N}(\tilde{M}_c(\tau))+\tilde{Q}_c(\tau+1/N)
$$
where $\tilde{Q}_c(\tau)=\frac{Q_c(\tau N)}{N}$,$Q_c(t+1)=M_c(t+1) - M_c(t) - \mathbb{E}[M_c(t+1) - M_c(t)|\mathbf{M}(t),B]$, $\mathcal{H}_{c,b_{c},N}(y)=H_{c,b_{c},N}(M_c(t)) \indicator{\underset{{k\in [C]}}{\max} H_{k,b_{k},N}(M_{k}(t))= H_{c,b_{c},N}(M_{c}(t))}\indicator{Nb_c>M_c(t)}$, and let $\tilde{\mathbf{M}}(\tau)=(\tilde{M}_1(\tau),\hdots,\tilde{M}_{C}(\tau))$.

When $N \to \infty$, and $t\in [T]$, the function $H_{c,b_{c},N}(M_c(t))$ converges to:
$$
    H_{c,b_{c},N}(M_c(t)) \to \sum_{d=1}^{D} \left( 1-e^{-a_{c,d}(b_c- \frac{M_c(t)}{N})} \right) \nu(d) = f_{c,b_{c}}(M_c(t)/N).
$$

Let $g(\tilde{M})=\left(f_{c,b_{c}}(\tilde{M}_c) \indicator{\underset{{k\in [C]}}{\max} f_{k,b_{k}}(\tilde{M}_k)= f_{c,b_{c}}(\tilde{M}_c)}\indicator{b_c > \tilde{M}_c} \right)_{i\in [C]}$ be the drift vector. According to Lemma \ref{lemma:assum_gast1}, each element of the drift vector satisfies the first assumption of Theorem 1 in \cite{gast:inria-00491859}. Moreover, according to Lemma \ref{lemma:martingale}, $\tilde{Q}_{c}(\tau+1/n)$ satisfies the second assumption of Theorem 1 in \cite{gast:inria-00491859}. Since $\tilde{M}_c(0)=0$, based on Theorem 1 in  \cite{gast:inria-00491859}, for all $\tau \in [T/N]$, $\tilde{M}(\tau)$ converges to $m(\tau)$, where $m$ is the solution of the following differential inclusion:
$$
    \dot{m}(\tau) \in G(m(\tau))
$$
where $G(m)= \operatorname{conv}\left(f_{c,b_{c}}(m_c) e_c  \mid c\in \arg\max_{j\in [C]} f_{j,b_{j}}(m_j) , b_c> m_c\right)$, and $\operatorname{conv}$ denotes the convex hull.

Following the same results as in \Cref{subsec:balance}, one can prove that $G$ is upper semicontinuous and $L$-one-sided Lipschitz, where $L=\max_{c\in [C]} L_c$. This ensures uniqueness of the solution $m$. To get the explicit rate of the convergence of $\tilde{M}_c$ to $m_c$, we use Theorem $4$ of \cite{gast:inria-00491859}. According to \Cref{lemma:martingale}, for $t\in [T]$, $\lE[|Q_c(t+1)|^2|\mathbf{M}(t),B]\leq b$ and $F$ is one sided Lipschitz with constant $L= \max_{c\in [C]} L_c$ where $L_c$ is defined in \Cref{lemma:fiLipschitz}. Thus according to theorem 4 in \cite{gast:inria-00491859}, taking $\delta_c= \sum_{d=1}^{D}\frac{a_{c,d}}{Ne}\nu(d)$, $K_{\alpha}$ as defined in \Cref{lemma:18gast} and $U_c= \sup_{0\leq t\leq T}f_{c,d_{c}}(M_c(t))$, taking $U_c= \sum_{d=1}^{D}(1-e^{-a_{c,d}b_c})\nu(d)$ we define $A_{\alpha,c}= U_c(U_c^2+\frac{14U_c}{3}+2K_{\alpha}), B_{\alpha,c}= 2U_c^2+4L\delta+12K_{\alpha}, C_{\alpha,c}= 2U_c^2+4L\epsilon +8K_{\alpha}$. We have, 
    \begin{align}
        \lP\left(\sup_{0\leq t\leq T} \left|\frac{M_c(t)}{N} -m_c(t/N)\right|\geq \min(\alpha, e^{L\alpha}/\sqrt{2L})\sqrt{A_{\alpha,c}/N+ \delta_c B_{\alpha,c}+\epsilon C_{\alpha,c}}\right)\leq \frac{b T}{N^2\epsilon}
    \end{align}
\end{proof}

\section{Proof of \Cref{theorem:regretetcbalance}}
In this section, we provide the proof of \Cref{theorem:regretetcbalance}. 
\regretetcbalance*
The proof of \Cref{theorem:regretetcbalance} is structured in two main steps:
\begin{itemize}
    \item We begin by establishing a concentration result for an estimator of \( \left(1 - \frac{a_{c,d}}{N}\right)^{N b_c - M_c(t)} \).
    \item Next, we decompose the regret and derive bounds for each resulting term.
\end{itemize}
Let $m',m \in [0, N b_c)$ and $\mathcal{V}_m=\{m'\in[0, N b_c) | 1/2\leq \frac{Nb_c-m'}{Nb_c-m}\leq 2\}$, we consider a Bernoulli random variable $Y_{c,d,m'}(t)$ with parameter $
1-D_{c,d}(m') := 1 - (1-a_{c,d}/N)^{Nb_c-m'}$, whenever $c(t) = c$, $d(t) = d$, and $M_c(t) = m'$. Let the number of  observations defined as $
T_{c,d,m'} := \sum_{t=1}^T \indicator{c(t)= c,\, d(t) = d,\, M_c(t) = m'\}}$ and $T_{{total}}=\sum_{m'\in \mathcal{V}_m} T_{c,d,m'}
$. We consider the following estimator:
\begin{align}
\label{eq:estimatorm_primem}
    \Theta(m):&= \frac{1}{T_{total}} \sum_{m'\in \mathcal{V}_m}\sum_{t=1}^T \indicator{c(t) = c,\, d(t) = d,\, M_c(t) = m'} (1 - Y_{c,d,m'}(t)) .\\
    \lE[\Theta(m)]&= \frac{1}{T_{total}}\sum_{m'\in \mathcal{V}_m}T_{c,d,m'} D_{c,d}(m')\\
    \lE[\Theta(m)]&=\frac{1}{T_{total}}\sum_{m'\in \mathcal{V}_m}T_{c,d,m'} D_{c,d}(m)^{\frac{Nb_c-m'}{Nb_c-m}} =g(D_{c,d}(m))
\end{align}
 with $D_{c,d}(m)=(1-a_{c,d}/N)^{Nb_c-m}$. The next lemma shows that $g^{-1}$ is Lipschitz continuous.  
\begin{lemma}
\label{lemma:g_lipchitz}
$g^{-1}$ is $2e^{a}$ Lipschitz. 
\end{lemma}
\begin{proof}
  Let 
\[
g(x)= \frac{1}{T_{{total}}} \sum_{m'\in \mathcal{V}_m} T_{c,d,m'} \, x^{\frac{Nb_c - m'}{Nb_c - m}},
\]
defined on the interval \( x \in \left[(1 - \tfrac{a}{N})^{Nb_c}, 1 \right] \), where \( a_{c,d} \leq a < N \). The function \( g \) is continuously differentiable and strictly increasing on this interval, as it is a finite sum of positive-coefficient power functions. Hence, \( g \) is invertible, and its inverse is differentiable with 
\[
\frac{\mathrm{d}g^{-1}}{\mathrm{d}y}(y) = \frac{1}{g'(g^{-1}(y))}.
\]

We compute
\[
g'(x) = \frac{1}{T_{{total}}} \sum_{m'\in \mathcal{V}_m} T_{c,d,m'} \cdot \frac{Nb_c - m'}{Nb_c - m} \cdot x^{\frac{Nb_c - m'}{Nb_c - m} - 1},
\]
so that
\[
\frac{\mathrm{d}g^{-1}}{\mathrm{d}y}(y) = \frac{1}{\frac{1}{T_{{total}}} \sum_{m'\in \mathcal{V}_m} T_{c,d,m'} \cdot \frac{Nb_c - m'}{Nb_c - m} \cdot (g^{-1}(y))^{\frac{Nb_c - m'}{Nb_c - m} - 1}}.
\]

To bound this derivative, we lower bound \( g'(x) \) over the domain. Let 
\[
\Delta := \min\left\{ 1, \left(1 - \tfrac{a}{N} \right)^{Nb_c \cdot \max_{m'\in \mathcal{V}_m} \left( \frac{Nb_c - m'}{Nb_c - m} - 1 \right)} \right\},
\]
then for all \( x \in \left[(1 - \tfrac{a}{N})^{Nb_c}, 1 \right] \),
\[
g'(x) \geq \frac{\Delta}{T_{{total}}} \sum_{m'\in \mathcal{V}_m} T_{c,d,m'} \cdot \frac{Nb_c - m'}{Nb_c - m}.
\]
This yields
\[
\frac{\mathrm{d}g^{-1}}{\mathrm{d}y}(y) \leq \frac{1}{\frac{\Delta}{T_{{total}}} \sum_{m'\in \mathcal{V}_m} T_{c,d,m'} \cdot \frac{Nb_c - m'}{Nb_c - m}}.
\]

Using the assumption \( \mathcal{V}_m = \left\{ m' \in [0, Nb_c) \,\middle|\, \frac{1}{2} \leq \frac{Nb_c - m'}{Nb_c - m} \leq 2 \right\} \)
we get,
\[
\frac{\mathrm{d}g^{-1}}{\mathrm{d}y}(y) \leq \frac{(1 - \tfrac{a}{N})^{-Nb_c}}{1/2} = 2(1 - \tfrac{a}{N})^{-Nb_c} \leq 2e^a,
\]
So \( g^{-1} \) is \( 2e^a \)-Lipschitz.

\end{proof}
The next result establishes a concentration result on an estimator of $ \left(1 - \frac{a_{c,d}}{N}\right)^{N b_c - M_c(t)} $.
\begin{lemma}
\label{lemma:concentrationDcd}
With probability $1-\delta$, $\hat{D}_{c,d}(m)=g^{-1}(\Theta(m))$ satisfies, 
\begin{align}
    |\hat{D}_{c,d}(m)-{D}_{c,d}(m)|\leq 2e^{a} \sqrt{\frac{\log(2/\delta)}{2T_{total}}}
\end{align}
\end{lemma}
\begin{proof} 
Let $\Theta(m)=\frac{1}{T_{total}} \sum_{m'\in \mathcal{V}_m}\sum_{t=1}^T \indicator{\{c(t) = c,\, d(t) = d,\, M_c(t) = m'\}} (1 - Y_{c,d,m'}(t))$, it is a sum of independent Bernoulli random variables, then using Hoeffding inequality,  with probability $1-\delta$, it satisfies,
    $$|\Theta(m)-\lE(\Theta(m))|\leq  \sqrt{\frac{\log(2/\delta)}{2T_{total}}}$$
As proved in \Cref{lemma:g_lipchitz}, $g^{-1}$ is $2e^{a}$ Lipschitz, $\hat{D}_{c,d}(m)=g^{-1}(\Theta(m))$ satisfies with probability $1-\delta$, 
\begin{align}
    |\hat{D}_{c,d}(m)-{D}_{c,d}(m)|\leq 2e^{a} \sqrt{\frac{\log(2/\delta)}{2T_{total}}}
\end{align}
\end{proof}

Let $R(T)$ be the cumulative regret defined by, 

\begin{align*}
    R(T)&=\sum_{c \in \mathcal{C}}M_c(T)-\hat{M}_c(T)\\ 
    &=R_{\text{explore}}(T_{\text{explore}})+ R_{\text{exploit}}(T_{\text{exploit}})
\end{align*}

We consider $T_{\text{explore}}=T^{\omega}$ with $0<\omega<1$ and $T_{\text{exploit}}=T-T_{\text{explore}}$, and $M_c$ is the matching size in the class $c$ created by $\balance$ and $\hat{M}_c$ is the matching size in the class $c$ created by $\learbalance$. Since $\learbalance$ may not be making optimal decision during exploration, we can bound $R_{\text{explore}}(T_{\text{explore}})$, as $R_{\text{explore}}(T_{\text{explore}})\leq C T^{\omega}$. For exploitation phase, let $m_c$ be the solution of differential inclusion defined in \Cref{theorem:inclusiondiffsol}. $R_{\text{exploit}}(T_{\text{exploit}})$ satisfies,
\begin{align}
    R_{\text{exploit}}(T_{\text{exploit}})= \sum_{c\in \mathcal{C}}\underbrace{\sum_{t=T_{\text{explore}}}^{T-1} \indicator{\exists n \in \mathcal{N}_c(t), c^*=c, m_n(t+1)=1} }_{M_c^{\text{exploit}}}- \underbrace{\sum_{t=T_{\text{explore}}}^{T-1} \indicator{\exists n \in \mathcal{N}_c(t), d^*=c, m_n(t+1)=1} }_{\hat{M}_c^{\text{exploit}}}
\end{align}
$c^*$ is the class chosen by $\balance$ and $d^*$ is the class chosen by $\learbalance$. 
\begin{align}
    &|R_{\text{exploit}}(T_{\text{exploit}})|\\
    &\leq \sum_{c\in \mathcal{C}} |{M}_c^{\text{exploit}}-Nm_c(T/N)+Nm_c(T/N)-\hat{M}_c^{\text{exploit}}|\\
    &\leq \sum_{c\in \mathcal{C}} |M_c(T)-M_c(T_{\text{explore}})-Nm_c(T/N)+Nm_c(T/N)-\hat{M}_c(T)+\hat{M}_c(T_{\text{explore}})|\\
    &\leq \sum_{c\in \mathcal{C}} \underbrace{|M_c(T)-Nm_c(T/N)|}_{w_{1}}+\underbrace{|\hat{M}_c(T_{\text{explore}})-M_c(T_{\text{explore}})|}_{w_{2}}+\underbrace{|Nm_c(T/N)-\hat{M}_c(T)|}_{w_{3}}
\end{align}
To bound $w_2$, we use the previous bound of $CT^{\omega}$. To bound \( w_1 \), we leverage the result from \Cref{theorem:inclusiondiffsol}. The next lemma, built upon \Cref{theorem:inclusiondiffsol}, quantifies the discrepancy between \( M_c \), the matching size produced by the \(\balance\) algorithm, and \( m_c \), the solution to the differential inclusion described therein.

\begin{lemma}
For $c\in \mathcal{C}$ and $0<q<1$, 
    $$|{M}_c(T)-Nm_c(T/N)| \sim \cO(T^{\frac{q+3}{4}})$$
\end{lemma}
\begin{proof}
  From \Cref{theorem:inclusiondiffsol}, with probability  $1-\frac{b\alpha}{N\epsilon^2}$, 
 \begin{align}
     \left|\frac{M_c(T)}{N} -m_c(T/N)\right|\\
       &\leq \min(\alpha, e^{L\alpha}/\sqrt{2L})\sqrt{A_{\alpha,c}/N+ \delta_c B_{\alpha,c}+\epsilon C_{\alpha,c}}
    \end{align}
Taking $\epsilon^{2}=\frac{1}{N^{1-q}}$ with $0<q<1$, 
$\delta_c= \sum_{d=1}^{D}\frac{a_{i,j}}{Ne}\nu(d)$, $K_{\alpha}=(c\alpha+\epsilon)e^{c\alpha}/c$ with  $c$ defined in \Cref{lemma:assum_gast1}, $U_c= \sum_{d=1}^{D}(1-e^{-a_{c,d}b_c})\nu(d) \leq 1$, $A_{\alpha,c}= U_c(U_c^2+\frac{14U_c}{3}+2K_{\alpha}) \leq (\frac{17}{3}+2K_{\alpha}), B_{\alpha,c}= 2U_c^2+4L\delta_c+12K_{\alpha}\leq 2+4L\delta_c+12K_{\alpha}, C_{\alpha,c}= 2U_c^2+4L\epsilon +8K_{\alpha}\leq 2+4L\epsilon+8K_{\alpha}$. 

\begin{align}
    \sqrt{A_{\alpha,c}/N+ \delta_c B_{\alpha,c}+\epsilon C_{\alpha,c}}&\leq \sqrt{\frac{(\frac{17}{3}+2K_{\alpha})}{N}+\delta_c(2+4L\delta_c+12K_{\alpha})+\epsilon(2+4L\epsilon+8K_{\alpha})}\\
    &\leq \sqrt{\frac{17/3+ \alpha e^{c\alpha}+2}{N}+\frac{e^{c\alpha}/c}{N^{\frac{3-q}{2}}}+\frac{4L}{N^2}+\frac{2+\alpha e^{c\alpha}}{N^{\frac{1-q}{2}}}+\frac{4L+e^{\alpha c}/c}{N^{1-q}}}
\end{align}
Thus $ \sqrt{A_{\alpha,c}/N+ \delta_c B_{\alpha,c}+\epsilon C_{\alpha,c}} \sim \cO(\frac{1}{N^{\frac{1-q}{4}}})$. Since $T= \alpha N$ with $\alpha>1$,  this leads to $ \left|{M_c(T)} -m_c(T/N)\right| \sim \cO(T^{\frac{q+3}{4}})$. 
\end{proof}

The next lemma is for bounding $w_3$, 
\begin{restatable}{lemma}
{boundhtaMiandmi}\label{theorem:boundhtaMiandmi}
    $$|Nm_c(T/N)-\hat{M}_c(T)| \sim \cO(T^{\frac{q+3}{4}}) $$
\end{restatable}
\begin{proof}

The goal here is to apply the differential inclusion approximation, to bound $|\hat{M}_c(T_{{exploit}})-Nm_c(T_{{exploit}}/N)|$. As defined previously, $\hat{M}_c$ is the matching size in the class $c\in \cal C$, built by a policy that considers the estimator $\hat{D}_{c,d}(m)$. The process $\hat{M}_c$ satisfies then for $t\in [T_{{explore}}+1,T]$, 
$$\hat{M}_c(t+1)=\hat{M}_c(t)+ \indicator{\exists n \in \mathcal{N}_c(t), c^*=c, m_n(t+1)=1}$$
Here $c^*$ is the class chosen by $\learbalance$. Let us compute the expected one step change of the process $\hat{M}_c$,
\begin{align}
    \lE[\hat{M}_c(t+1)-\hat{M}_c(t)|\hat{M}_c(t),B]&= \hat{H}_{c,b_{c},N}(\hat{M}_c(t))\indicator{\underset{{k\in [C]}}{\max} \hat{H}_{k,b_{k},N}(\hat{M}_{k}(t))= \hat{H}_{c,b_{c},N}(\hat{M}_{c}(t))}\\
    &=\hat{\mathcal{H}}_{c,b_{c},N}(\hat{M}_c(t))
\end{align}
where $\hat{H}_{c,b_{c},N}(\hat{M}_c(t))=\sum_{d=1}^{D}\left(1-\hat{D}_{c,d}(\hat{M}_c(t))\right) \nu(d)$. Let $\hat{Q}_c(t+1)=\hat{M}_c(t+1)-\hat{M}_c(t)- \lE[\hat{M}_c(t+1)-\hat{M}_c(t)| B, \hat{M}_c(t)]$. Here $Q_c$ is a martingale difference sequences that satisfies the same assumptions in \Cref{lemma:martingale}. for $t \in [T_{explore}+1,T]$
\begin{align}
\label{eq:evolutionhatMc}
\hat{M}_c(t+1)&= \hat{M}_c(t)+\hat{\mathcal{H}}_{c,b_{c},N}(\hat{M}_c(t))+ \hat{Q}_c(t+1)\\
&= \hat{M}_c(t)+{\mathcal{H}}_{c,b_{c},N}(\hat{M}_c(t))+ \Delta_{c,b_{c},N}(\hat{M}_c(t)) + \hat{Q}_c(t+1)
\end{align}
Note that the function \( \mathcal{H}_{c, b_c, N}(\hat{M}_c(t)) \) is the same as the one defined in \Cref{eq:processmatchdiffinclu}. The goal of the proof is to show that $\frac{1}{N} \sum_{l=1}^{t} \Delta_{c, b_c, N}(\hat{M}_c(l))$ converges to $0$
with high probability. This is important because, according to the proofs of Theorems~1 and 4 in~\cite{gast:inria-00491859}, establishing this convergence implies that the process \( \hat{M}_c \) converges to the same solution as the differential inclusion introduced in \Cref{theorem:inclusiondiffsol}. This follows from the fact that we can write the evolution of \( \hat{M}_c \) as in \Cref{eq:evolutionhatMc}
where \( \hat{Q}_c \) is a martingale difference term and \( \mathcal{H}_{c, b_c, N}(\hat{M}_c(t)) \) is the drift of the process $M_c$ defined in \Cref{eq:processmatchdiffinclu}. Thus, showing that the average of the perturbation term \( \Delta_{c, b_c, N}(\hat{M}_c(t)) \) vanishes ensures that \( \hat{M}_c \) asymptotically follows the same differential inclusion introduced in \Cref{theorem:inclusiondiffsol}.

We now turn our attention to analyzing \( \Delta_{c, b_c, N}(\hat{M}_c(t)) \),
\begin{align}
    \Delta_{c,b_{c},N}(\hat{M}_c(t))&= \hat{H}_{c,b_{c},N}(\hat{M}_c(t))\indicator{\underset{{k\in [C]}}{\max} \hat{H}_{k,b_{k},N}(\hat{M}_{k}(t))= \hat{H}_{c,b_{c},N}(\hat{M}_{c}(t))}\\
    &- {H}_{c,b_{c},N}(\hat{M}_c(t))\indicator{\underset{{k\in [C]}}{\max} {H}_{k,b_{k},N}(\hat{M}_{k}(t))={H}_{c,b_{c},N}({M}_{c}(t))}
\end{align}
\begin{align}
    &|\Delta_{c,b_{c},N}(\hat{M}_c(t))|
    \leq \left|\hat{H}_{c,b_{c},N}(\hat{M}_c(t))-{H}_{c,b_{c},N}(\hat{M}_c(t))\right| \\
    &+{H}_{c,b_{c},N}(\hat{M}_c(t))\left|\indicator{\underset{{k\in [C]}}{\max} \hat{H}_{k,b_{k},N}(\hat{M}_{k}(t))= \hat{H}_{c,b_{c},N}(\hat{M}_{c}(t))}-\indicator{\underset{{k\in [C]}}{\max} {H}_{k,b_{k},N}(\hat{M}_{k}(t))={H}_{c,b_{c},N}({M}_{c}(t))}\right|
\end{align}
From the concentration inequality in \Cref{lemma:concentrationDcd}, we have, 
\begin{align}
    A=\left|\hat{H}_{c,b_{c},N}(\hat{M}_c(t))-{H}_{c,b_{c},N}(\hat{M}_c(t))\right| &\leq \sum_{d=1}^{D} |\hat{D}_{c,d}(\hat{M}_c(t))-{D}_{c,d}(\hat{M}_c(t))|\nu(d)
\end{align}
with probability at least $1-\delta$,
\begin{align}
    A\leq 2e^{a} \sqrt{\frac{\log(2/\delta)}{2T_{total}}}
\end{align}
 Now let us focus on $\left|\indicator{\underset{{k\in [C]}}{\max} \hat{H}_{k,b_{k},N}(\hat{M}_{k}(t))= \hat{H}_{c,b_{c},N}(\hat{M}_{c}(t))}-\indicator{\underset{{k\in [C]}}{\max} {H}_{k,b_{k},N}(\hat{M}_{k}(t))={H}_{c,b_{c},N}({M}_{c}(t))}\right|$ we need to bound the mismatch between the indicator functions.
These indicators differ only when the maximum changes, i.e., when the "argmax" under $\hat{D}_{c,d}(\hat{M}_c(t))$ and ${D}_{c,d}(\hat{M}_c(t))$ do not agree. Suppose ${H}_{c,b_{c},N}(\hat{M}_c(t))$ is the largest value among all ${H}_{k,b_{k},N}(\hat{M}_k(t))$ by a margin $\gamma >0$, 
$${H}_{c,b_{c},N}(\hat{M}_c(t))> {H}_{k,b_{k},N}(\hat{M}_k(t))+\psi ~~~~~~\text{for}~k\neq c$$
To ensure that the estimator $\hat{D}_{c,d}$ does not flip the argmax, we want to control how much each function ${H}_{k,b_{k},N}(\hat{M}_k(t))$ can change under small perturbations in $D_{c,d}$.
Suppose that the maximizer changes for the function $\hat{H}$, this means that for some $j\neq c$, 
\begin{align}
    \hat{H}_{j,b_{j},N}(\hat{M}_j(t))(\hat{M}_j(t))\geq \hat{H}_{c,b_{c},N}(\hat{M}_c(t))
\end{align}
Thus, 

\begin{align}
     \hat{H}_{j,b_{j},N}(\hat{M}_j(t))- \hat{H}_{c,b_{c},N}(\hat{M}_c(t))&= \hat{H}_{j,b_{j},N}(\hat{M}_j(t))-{H}_{j,b_{j},N}(\hat{M}_j(t))-\hat{H}_{c,b_{c},N}(\hat{M}_c(t))\\
     &+{H}_{c,b_{c},N}(\hat{M}_c(t))+ {H}_{j,b_{j},N}(\hat{M}_j(t))-{H}_{c,b_{c},N}(\hat{M}_c(t))
\end{align}
$\hat{H}_{j,b_{j},N}(\hat{M}_j(t))\geq \hat{H}_{c,b_{c},N}(\hat{M}_c(t))$ implies,
\begin{align}
      J=&\hat{H}_{j,b_{j},N}(\hat{M}_j(t))-{H}_{j,b_{j},N}(\hat{M}_j(t))-\hat{H}_{c,b_{c},N}(\hat{M}_c(t))+{H}_{c,b_{c},N}(\hat{M}_c(t))\\
      J \geq& {H}_{c,b_{c},N}(\hat{M}_c(t))-{H}_{j,b_{j},N}(\hat{M}_j(t))\geq \psi
\end{align}  
Applying the triangle inequality, we get, 
\begin{align}
    B&=|\hat{H}_{j,b_{j},N}(\hat{M}_j(t))-{H}_{j,b_{j},N}(\hat{M}_j(t))-\hat{H}_{c,b_{c},N}(\hat{M}_c(t))+{H}_{c,b_{c},N}(\hat{M}_c(t))|
    \\
    B&\leq |\hat{H}_{j,b_{j},N}(\hat{M}_j(t))-{H}_{j,b_{j},N}(\hat{M}_j(t))|+|\hat{H}_{c,b_{c},N}(\hat{M}_c(t))-{H}_{c,b_{c},N}(\hat{M}_c(t))|
\end{align}
Thus by the margin condition, we have 
\begin{align}
    |\hat{H}_{j,b_{j},N}(\hat{M}_j(t))-{H}_{j,b_{j},N}(\hat{M}_j(t))|+|\hat{H}_{c,b_{c},N}(\hat{M}_c(t))-{H}_{c,b_{c},N}(\hat{M}_c(t))|\geq \psi\\ 
    \Rightarrow \max\{|\hat{H}_{j,b_{j},N}(\hat{M}_j(t))-{H}_{j,b_{j},N}(\hat{M}_j(t))|+|\hat{H}_{c,b_{c},N}(\hat{M}_c(t))-{H}_{c,b_{c},N}(\hat{M}_c(t))|\}\geq \psi/2
\end{align}
Let $\gamma_t= \left|\indicator{\underset{{k\in [C]}}{\max} \hat{H}_{k,b_{k},N}(\hat{M}_{k}(t))= \hat{H}_{c,b_{c},N}(\hat{M}_{c}(t))}-\indicator{\underset{{k\in [C]}}{\max} {H}_{k,b_{k},N}(\hat{M}_{k}(t))={H}_{c,b_{c},N}({M}_{c}(t))}\right|$, according to previous development, we deduce that, 
\begin{align}
    \gamma_t\leq \indicator{\exists k \in \mathcal{C}, |\hat{H}_{k,b_{k},N}(\hat{M}_k(t))-{H}_{k,b_{k},N}(\hat{M}_k(t))|\geq \psi/2}
\end{align}
But from concentration, we have the following condition on $\gamma$, 
\begin{align}
   \psi\geq  4e^{a} \sqrt{\frac{\log(2/\delta)}{2T_{total}}}
\end{align}
Thus choosing $\delta$ such that this condition is satisfied, we ensure that  with high probability $\frac{1}{N}\sum_{l=1}^{t} \Delta_{c,b_{c},N}(\hat{M}_c(l))$ tends to $0$. Having established this convergence, and noting that the process $\hat{M}_c$ can be represented in the form given in \Cref{eq:processmatchdiffinclu}. Thus according to the proof of Theorem 1 and 4 in \cite{gast:inria-00491859},  $\hat{M}_c$ converges to the solution of the same differential inclusion defined in \Cref{theorem:inclusiondiffsol}. and we have $|\hat{M}_c(T)-Nm_c(T/N)| \sim \cO(N^{\frac{q+3}{4}})$ for $0<q<1$.

\end{proof}
With all the previous result, $R(T)=R_{\text{explore}}(T_{\text{explore}})+ R_{\text{exploit}}(T_{\text{exploit}})$, we showed first that $R_{\text{explore}}(T_{\text{explore}}) \sim \cO(T^{\omega})$, and $R_{\text{exploit}}(T_{\text{exploit}})\leq C (2T^{\frac{q+3}{4}}+ T^{\omega})$, thus choosing $\omega=\frac{q+3}{4}$, we get that $R(T)\sim\cO(T^{\frac{q+3}{4}})$.
\end{document}